\documentclass[
final
]{dmtcs-episciences}

\usepackage{amsmath}
\usepackage{amssymb}

\newenvironment{proof}{\noindent{\bf Proof:}}
{\begin{flushright} \vspace{-0.5cm} $\square$ \end{flushright}}

\newcommand {\abs}[1]{\left\vert#1\right\vert}
\newcommand {\set}[1]{\left\{#1\right\}}

\newcommand {\defined} {\stackrel{def} {=}}

\newtheorem{observation} {Observation}  [section]

\newtheorem{theorem}  {Theorem}         [section]
\newtheorem{lemma}    {Lemma}           [section]

\usepackage{algorithm}
\usepackage{algorithmicx}
\usepackage[noend]{algpseudocode}
\usepackage{algpascal}
\usepackage[round]{natbib}

\newcommand{\runningtitle}[1]{\vspace{0.5ex}\noindent{{\small \textbf{\boldmath #1~}}}\\}

\newcommand{\nph}{\textsc{NP}\textrm{-Hard}}
\newcommand{\npc}{\textsc{NP}\textrm{-Complete}}
\newcommand{\eptg}[1] {\textsc{Ept}(#1)}
\newcommand{\vptg}[1] {\textsc{Vpt}(#1)}
\newcommand{\enptg}[1] {\textsc{Enpt}(#1)}
\newcommand{\epg} {\textsc{EPG}}
\newcommand{\ept} {\textsc{EPT}}
\newcommand{\vpt} {\textsc{VPT}}
\newcommand{\enpt} {\textsc{ENPT}}
\newcommand{\pp} {{\cal P}}
\newcommand{\eptgp} {\eptg{\pp}}
\newcommand{\vptgp} {\vptg{\pp}}
\newcommand{\enptgp} {\enptg{\pp}}
\newcommand{\rep} {\left< T,\pp \right>}
\newcommand{\repn}[1] {\left< T_{#1},\pp_{#1} \right>}
\newcommand{\ppbar} {\bar{\pp}}
\newcommand{\repbar} {\left< \bar{T},\ppbar \right>}
\newcommand{\repbarprime} {\left< \bar{T'},\ppbar' \right>}

\newcommand{\repprime} {\left< T',\pp' \right>}
\newcommand{\repdprime} {\left< T'',\pp'' \right>}

\renewcommand{\split} {\textit{split}}
\newcommand{\op} {{\cal O}}
\newcommand{\opg} {\op(G,C)}

\newcommand{\wdt} {{\cal W}}
\newcommand{\wdtg} {\wdt(G,C)}

\newcommand{\wdtgdprime} {\wdt(G'',C'')}

\newcommand{\contract}[2] {{#1}_{/{#2}}}
\newcommand{\contractge} {\contract{G}{e}}
\newcommand{\contractgce} {\contract{(G,C)}{e}}
\newcommand{\contractp}[1] {\contract{\rep}{#1}}
\newcommand{\contractppq}{\contractp{P_p,P_q}}
\newcommand{\kp} {(K_4,P_4)}
\newcommand{\aggressive}[2] {\contract{#1}{#2}}

\newcommand{\aggressivegcdpk} {\aggressive{(G'',C'')}{K}}
\newcommand{\comp} {\textit{comp}}

\newcommand{\xx} {{\cal X}}
\newcommand{\yy} {{\cal Y}}
\newcommand{\prb} {\textsc{HamiltonianPairRec}}
\newcommand{\prbpthree} {\textsc{P3-HamiltonianPairRec}}

\title{Graphs of Edge-Intersecting Non-Splitting Paths in a Tree: Representations of Holes-Part II\thanks{This work was supported in part by the Israel Science Foundation grant No. 1249/08, by TUBITAK PIA BOSPHORUS Grant No. 111M303 and by the TUBITAK 2221 Programme.}
\thanks{A preliminary version of this paper was presented in WG 2013 (\cite{BESZ13}).}
}

\author{Arman Boyac{\i} \affiliationmark{1}
  \and T{\i}naz Ekim \affiliationmark{1}
  \and Mordechai Shalom \affiliationmark{1,2} \thanks{Part of this work is accomplished while this author was visiting Bogazici University, Department of Industrial  Engineering, under the TUBITAK 2221 Program whose support is greatly acknowledged.}
  \and Shmuel Zaks \affiliationmark{3}}

\affiliation{
  Department of Industrial Engineering, Bogazici University, Istanbul, Turkey \\
  TelHai College, Upper Galilee, 12210, Israel\\
  Department of Computer Science, Technion, Haifa, Israel
}
\keywords{Intersection Graphs, Path Graphs, EPT Graphs}

\received{2016-8-14}
\revised{2017-10-3}
\accepted{2017-12-19}

\begin{document}
\publicationdetails{20}{2018}{1}{2}{3974}
\maketitle
\begin{abstract}
Given a tree and a set $\pp$ of non-trivial simple paths on it, $\vptgp$ is the $\vpt$ graph (i.e. the vertex intersection graph) of the paths $\pp$, and $\eptgp$ is the $\ept$ graph (i.e. the edge intersection graph) of $\pp$. These graphs have been extensively studied in the literature.
Given two (edge) intersecting paths in a graph, their \emph{split vertices} is the set of vertices having degree at least $3$ in their union. A pair of (edge) intersecting paths is termed \emph{non-splitting} if they do not have split vertices (namely if their union is a path).
We define the graph $\enptgp$ of edge intersecting non-splitting paths of a tree, termed the $\enpt$ graph, as the graph having a vertex for each path in $\pp$, and an edge between every pair of vertices representing two paths that are both edge-intersecting and non-splitting. A graph $G$ is an $\enpt$ graph if there is a tree $T$ and a set of paths $\pp$ of $T$ such that $G=\enptgp$, and we say that $\rep$ is a \emph{representation} of $G$.

Our work follows Golumbic and Jamison's research, in which they defined the $\ept$ graph class, and characterized the representations of chordless cycles (holes). Our main goal is the characterization of the representations of chordless $\enpt$ cycles. To achieve this goal, we assume that the $\ept$ graph induced by the vertices of an  $\enpt$ hole is given. We use the results of that research as building blocks in order to discover this characterisation, which turn out to have a more complex structure than in the case of $\ept$ holes. In the first part of this work we have shown that cycles, trees and complete graphs are $\enpt$ graphs. We also introduced three assumptions $(P1)$, $(P2)$, $(P3)$ defined on $\ept$, $\enpt$ pairs of graphs, and characterized the representations of $\enpt$ holes that satisfy $(P1)$, $(P2)$, $(P3)$. In this work we relax two of these three assumptions and characterize the representations of $\enpt$ holes satisfying $(P3)$. These two results are achieved by providing polynomial-time algorithms. Last we show that the problem of finding such a representation is $\nph$ in general, i.e. without assumption $(P3)$. This result extends in some sense the $\textsc{NP}$-Hardness of $\ept$ graph recognition shown in Golumbic and Jamison's work.

\end{abstract}

\section{Introduction}\label{sec:intro}
\subsection{Background}
Given a tree $T$ and a set $\pp$ of non-trivial simple paths in $T$, the Vertex Intersection Graph of Paths in a Tree (\vpt) and  the Edge Intersection Graph of Paths in a Tree ($\ept$) of $\pp$ are denoted by $\vptgp$ and $\eptgp$, respectively. Both graphs have $\pp$ as vertex set. $\vptgp$ (resp. $\eptgp$) contains an edge between two vertices if the corresponding two paths intersect in at least one vertex (resp. edge). A graph $G$ is $\vpt$ (resp. $\ept$) if there exist a tree $T$ and a set $\pp$ of non-trivial simple path in $T$ such that $G$ is isomorphic to $\vptgp$ (resp. $\eptgp$). In this case we say that $\rep$ is a $\vpt$ (resp. an $\ept$) representation of $G$.

In this work we focus on edge intersections of paths, therefore whenever we are concerned with intersection of paths we omit the word "edge" and simply write that two paths \emph{intersect}. The graph of edge intersecting and non-splitting paths of a tree ($\enpt$) of a given representation $\rep$ as described above, denoted by $\enptgp$, has a vertex $v$ for each path $P_v$ of $\pp$ and two vertices $u,v$ of $\enptgp$ are adjacent if the paths $P_u$ and $P_v$ intersect and do not split (that is, their union is a path). A graph $G$ is an $\enpt$ graph if there is a tree $T$ and a set of paths $\pp$ of $T$ such that $G$ is isomorphic to $\enptgp$. We note that $\eptgp=\enptgp$ is an interval graph whenever $T$ is a path. Therefore, the class $\enpt$ includes all interval graphs.

$\ept$ and $\vpt$ graphs have applications in communication networks. Consider a communication network of a tree topology $T$. The message routes to be delivered in this communication network are paths on $T$. Two paths conflict if they both require to use the same link (node). This conflict model is equivalent to an $\ept$ (a $\vpt$) graph. Suppose we try to find a schedule for the messages such that no two messages sharing a link (node) are scheduled in the same time interval. Then a vertex coloring of the $\ept$ ($\vpt$) graph corresponds a feasible schedule on this network.

$\ept$ graphs also appear in all-optical telecommunication networks. The so-called Wavelength Division Multiplexing (WDM) technology can multiplex different signals onto a single optical fiber by using different wavelength ranges of the laser beam (\cite{CGK92,R93}). WDM is a promising technology enabling us to deal with the massive growth of traffic in telecommunication networks, due to applications such as video-conferencing, cloud computing and distributed computing (\cite{DV93}). A stream of signals traveling from its source to its destination in optical form is termed a \emph{lightpath}. A lightpath is realized by signals traveling through a series of fibers, on a certain wavelength. Specifically, Wavelength Assignment problems (WLA) are a family of path coloring problems that aim to assign wavelengths (i.e. colors) to lightpaths, so that no two lightpaths with a common edge receive the same wavelength and a certain objective function (depending on the problem) is minimized. \emph{Traffic Grooming} is the term used for combination of several low-capacity requests into one lightpath using Time Division Multiplexing (TDM) technology (\cite{GRS98}). In this context a set of  paths can be combined into one lightpath, thus receiving the same color, as long as they satisfy the following two conditions:
\begin{itemize}
\item The load condition: at most $g$ lightpaths using the same fiber may receive the same color, where $g$ is an integer termed \emph{the grooming factor}.
\item  The union of the paths receiving the same color should constitute a path or a set of disjoint paths.
\end{itemize}
It follows that a feasible solution of the traffic grooming problem is a vertex coloring of the graph $\eptgp \setminus \enptgp$ where each color class induces a sub-graph of $\eptgp$ with clique number at most $g$. Therefore, it makes sense to analyze the structure of these graph pairs, i.e. the two graphs $\eptgp$ and $\enptgp$ defined on the same set of vertices.

\subsection{Related Work}
$\ept$ and $\vpt$ graphs have been extensively studied in the literature. Although $\vpt$ graphs can be characterized by a fixed number of forbidden subgraphs (\cite{Leveque2009Characterizing}), \cite{Golumbic1985151} showed that $\ept$ graph recognition is $\npc$. Edge intersection and vertex intersection give rise to identical graph classes in the case of paths in a line and in the case of subtrees of a tree. However, $\vpt$ graphs and $\ept$ graphs are incomparable in general; neither $\vpt$ nor $\ept$ contains the other. Main optimization and decision problems such as recognition, maximum clique, minimum vertex coloring and the maximum stable set are shown to be polynomial-time solvable for $\vpt$ graphs by \cite{Fanica1978211}, \cite{Fanica2000181}, \cite{Golumbic:2004:AGT:984029}  and \cite{Spinrad1995181}, respectively. The recognition and minimum vertex coloring problems are shown to remain $\npc$ for $\ept$ graphs by \cite{Golumbic1985151,GJ85}. In contrast, one can solve in polynomial time the maximum clique and the maximum stable set problems in $\ept$ graphs (\cite{GJ85} and \cite{RobertE1985221}, respectively).

After these works on $\ept$ and $\vpt$ graphs in the early 80's, this topic did not attract much further study until very recently. Current research on intersection graphs is concentrated on the comparison of various intersection graphs of paths in a tree and their relation to chordal and weakly chordal graphs (\cite{Golumbic2008Equivalences,GolumbicLS08}). Also, a tolerance model is studied in \cite{Golumbic2008Kedge} via $k$-edge intersection graphs where two vertices are adjacent if their corresponding paths intersect on at least $k$ edges . Besides, several recent papers consider the edge intersection graphs of paths on a grid (e.g. \cite{BiedlS10}).

Another related line of research is subsets of $\epg$ (i.e., edge-intersection of paths in a grid) graphs. Since every graph is an $\epg$ graph, research in this field focuses on subsets of this class of graphs that are obtained by imposing some restrictions on the paths in their representation. The most researched such restriction is an upper bound on the number of bends on every single path (e.g. \cite{GolumbicLS09}, \cite{HKU10}, \cite{EGM2013}). The recent work of \cite{CCH16} considers the subsets of one bend path obtained by restricting their orientation in the plane. \cite{PR16} focus on graphs admitting $\epg$ representations by paths with at most 2 bends and show that their recognition along with some subclasses of them is $\npc$.

\subsection{Our Contribution}
In the first part of this work (\cite{BESZ13-ENPT1-DAM}), we define the family of $\enpt$ graphs, and investigate its basic properties. We show that trees, cliques and cycles are $\enpt$ graphs. In the same work we started the study of the characterization of $\enpt$ representations of holes. To achieve this goal we assume that the input is augmented by the underlying $\ept$ graph. More particularly, we study pairs $(G,C)$ of graphs where $G=\eptgp$ and $C=\enptgp$ for some representation $\rep$. For any given such pair of graphs $(G,C)$ that satisfy three assumptions $(P1)$, $(P2)$, $(P3)$ and $C$ is a Hamiltonian cycle of $G$, we characterize the minimal representation of $(G,C)$ as a planar tour of the weak dual of $G$. In this work we extend these results. Namely a) we provide an algorithm that returns the minimal representation of a given pair of graphs $(G,C)$ that satisfy $(P3)$ and $C$ is a Hamiltonian cycle of $G$, b) using this algorithm we characterize the representations of such pairs, and c) we show that for a given pair $(G,C)$ in general, i.e. when $(G,C)$ does not necessarily satisfy $(P3)$, it is $\nph$ to decide whether $(G,C)$ has a representation.

Our approach can be summarized as follows. Given a pair $(G,C)$, we first remove $\kp$ pairs and contract all contactable $\enpt$ edges. A contraction operation corresponds to a union of two paths in the representation. Then we study the reverse operations of contractions (union in the representation) and characterization of representations of $\kp$s. In Section \ref{sec:prelim} we give definitions and preliminaries. In Section \ref{sec:uncontraction} we present basic results related to contraction, describe an algorithm returning minimal representations of pairs satisfying assumptions $(P2)$ and $(P3)$. Using this algorithm we characterize these representations. In Section \ref{sec:K4P4} we characterize the representations of $\kp$ pairs, we define an aggressive contraction operation and we present an algorithm that returns the minimal representation of a given pair $(G,C)$ satisfying $(P3)$, through which we characterize the representations. In Section \ref{sec:General} we show that when $(P3)$ does not hold, there does not exist a polynomial-time algorithm to find a representation of a given pair $(G,C)$, unless $\textsc{P}=\textsc{NP}$.

\section{Preliminaries}\label{sec:prelim}
In this section we provide definitions used in the paper, present known results related to our work, and develop basic results. The section is organized as follows: Section \ref{subsec:prelim-definitions} is devoted to basic definitions. In Section \ref{subsec:prelim-eptn-knownresults} we present known results on $\ept$ graphs that are closely related to our work. Section \ref{subsec:eptn-OurPreviousResults} summarizes known results related to $\enpt$ graphs and pairs.

\subsection{Definitions}\label{subsec:prelim-definitions}
\runningtitle{Notation:}
Given a graph $G$ and a vertex $v$ of $G$, we denote by $\delta_G(v)$ the set of edges of $G$ incident to $v$, by $N_G(v)$ the set consisting of $v$ and its neighbors in $G$, and by $d_G(v)=\abs{\delta_G(v)}$ the degree of $v$ in $G$. A vertex is called a \emph{leaf} (resp. \emph{intermediate vertex}, \emph{junction}) if $d_G(v)=1$ (resp. $=2$, $\geq 3$). Whenever there is no ambiguity we omit the subscript $G$ and write $\delta(v)$, $d(v)$, and $N(v)$.

Given a graph $G$, $\bar{V} \subseteq V(G)$ and $\bar{E} \subseteq E(G)$ we denote by $G[\bar{V}]$ and $G[\bar{E}]$ the subgraphs of $G$ induced by $\bar{V}$ and by $\bar{E}$, respectively.

The \emph{union} of two graphs $G, G'$ is the graph $G \cup G' \defined (V(G) \cup V(G'), E(G) \cup E(G'))$. The \emph{join} $G+G'$ of two disjoint graphs $G,G'$ is the graph $G \cup G'$ together with all the edges joining $V(G)$ and $V(G')$, i.e. $G + G' \defined (V(G) \cup V(G'), E(G) \cup E(G') \cup (V(G) \times V(G')))$.

Given a (simple) graph $G$ and $e \in E(G)$, we denote by $\contractge$ the (simple) graph obtained by contracting the edge $e$ of $G$, i.e. by coinciding the two endpoints of $e=\set{p,q}$ to a single vertex $p.q$, and then removing self loops and parallel edges. Let $\bar{E}=\set{e_1, e_2, \ldots e_k} \subseteq E(G)$. We denote by $\contract{G}{e_1,\ldots,e_k}$ the graph obtained from $\contract{G}{e_1,\ldots,e_{k-1}}$ by contracting the (image of the) edge $e_k$. The effect of such a sequence of contractions is equivalent to contracting every connected component of $G[\set{e_1,\ldots,e_k}]$ to a vertex. Therefore, the order of contractions is not important, i.e. for any permutation $\pi$ of $\set{1,\ldots,k}$ we have $\contract{G}{e_1,\ldots,e_{k-1}}=\contract{G}{e_{\pi(1)},\ldots,e_{\pi(k-1)}}$. Based on this fact, we denote by $\contract{G}{\bar{E}}$ the graph obtained by contracting the edges of $\bar{E}$ (in any order).

For two vertices $u,v$ of a tree $T$ we denote by $p_T(u,v)$ the unique path between $u$ and $v$ in $T$. A vertex $w$ of a path $P$ that is not an endpoint of $P$ is termed an \emph{internal vertex} of $P$. We also say that $P$ \emph{crosses} $w$. A \emph{cherry} of a tree $T$ is a connected subgraph of $T$ consisting of two leaves of $T$ adjacent to an internal vertex of $T$.

\runningtitle{Intersections and union of paths:}
Given two paths $P,P'$ in a graph, we write $P \parallel P'$ to denote that $P$ and $P'$ are non-intersecting, i.e. \emph{edge-disjoint}. The \emph{split vertices} of $P$ and $P'$ is the set of junctions in their union $P \cup P'$ and is denoted by $\split(P,P')$. Whenever $P$ and $P'$ intersect and $\split(P, P') = \emptyset$  we say that $P$ and $P'$ are \emph{non-splitting} and denote this by $P \sim P'$. In this case $P \cup P'$ is a path or a cycle. When $P$ and $P'$ intersect and $\split(P, P') \neq \emptyset$ we say that they are \emph{splitting} and denote this by $P \nsim P'$.  Clearly, for any two paths $P$ and $P'$ exactly one of the following holds: $P \parallel P'$, $P \sim P'$, $P \nsim P'$.

When the graph $G$ is a tree, the union $P \cup P'$ of two intersecting paths $P,P'$ on $G$ is a tree with at most two junctions, i.e. $\abs{\split(P,P')} \leq 2$ and $P \cup P'$ is a path whenever $P \sim P'$.

\runningtitle{The VPT, EPT and ENPT graphs:}
Let $\pp$ be a set of paths in a tree $T$. The graphs $\vptgp, \eptgp$ and $\enptgp$ are graphs such that $V(\enptgp) = V(\eptgp) = V(\vptgp)=\set{p | P_p \in \pp}$. Given two distinct paths $P_p,P_q \in \pp$, $\set{p,q}$ is an edge of $\enptgp$ if $P_p \sim P_q$, and $\set{p,q}$ is an edge of $\eptgp$ (resp. $\vptgp$) if $P_p$ and $P_q$ have a common edge (resp. vertex) in $T$. See Figure \ref{fig:vpt-ept-eptn} for an example. From these definitions it follows that
\begin{observation}
$E(\enptgp) \subseteq E(\eptgp)\subseteq E(\vptgp)$.
\end{observation}

\begin{figure}[htbp]
\centering
\includegraphics[width=\textwidth]{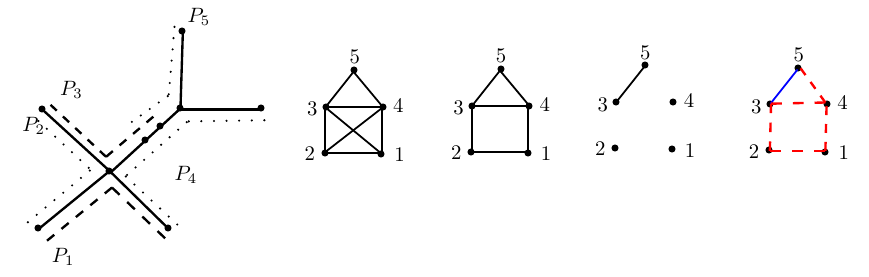}
\caption{A host tree $T$, a collection of paths $\pp = \left\{ P_1, P_2, P_3, P_4, P_5 \right\} $ defined on $T$ and the corresponding graphs $\vptgp, \eptgp$ and $\enptgp$. The last sub-figure shows the graphs $\eptgp$ and $\enptgp$ where their common edge is depicted as a solid line.}
\label{fig:vpt-ept-eptn}
\end{figure}

Two graphs $G$ and $G'$ such that $V(G)=V(G')$ and $E(G') \subseteq E(G)$ are termed a \emph{pair (of graphs)} denoted as $(G,G')$. If $\eptgp=G$ (resp. $\enptgp=G$) we say that $\rep$ is an $\ept$ (resp. $\enpt$) representation for $G$. If $\eptgp=G$ and $\enptgp=G'$ we say that $\rep$ is a representation for the pair $(G,G')$. Given a pair $(G,G')$ the sub-pair induced by $\bar{V} \subseteq V(G)$ is the pair $(G[\bar{V}],G'[\bar{V}])$. Clearly, any representation of a pair induces representations for its induced sub-pairs, i.e. the pairs have the hereditary property.

A \emph{cherry} of a representation $\rep$ is a cherry of $T$ with leaves $v,v'$ such that $v$ (resp. $v'$) is an endpoint of exactly one path $P$ (resp. $P'$) of $\pp$, and $P \neq P'$.

Throughout this work, in all figures, the edges of the tree $T$ of a representation $\rep$ are drawn as solid edges whereas the paths on the tree are shown by dashed, dotted, etc. edges. Similarly, edges of $\enptgp$ are drawn with solid or blue lines whereas edges in $E(\eptgp) \setminus E(\enptgp)$ are dashed or red. We sometimes refer to them as blue and red edges respectively. For an edge $e=\set{p,q}$ we use $\split(e)$ as a shorthand for $\split(P_p,P_q)$. We note that $e$ is a red edge if and only if $\split(e) \neq \emptyset$.

\runningtitle{Cycles, Chords, Holes, Outerplanar graphs, Weak dual trees:}
Given a graph $G$ and a cycle $C$ of it, a \emph{chord} of $C$ in $G$ is an edge of $E(G) \setminus E(C)$ connecting two vertices of $V(C)$. The \emph{length} of a chord connecting the vertices $i$,$j$ is the length of a shortest path between $i$ and $j$ on $C$. $C$ is a \emph{hole} (chordless cycle) of $G$ if $G$ does not contain any chord of $C$. This is equivalent to saying that the subgraph $G[V(C)]$ of $G$ induced by the vertices of $C$ is a cycle. For this reason a chordless cycle is also called an \emph{induced} cycle.

An \emph{outerplanar} graph is a planar graph that can be embedded in the plane such that all its vertices are on the unbounded face of the embedding. An outerplanar graph is Hamiltonian if and only if it is biconnected; in this case the unbounded face forms the unique Hamiltonian cycle. The \emph{weak dual} graph of a planar graph $G$ is the graph obtained from its dual graph, by removing the vertex corresponding to the unbounded face of $G$. The weak dual graph of an outerplanar graph is a forest, and in particular the weak dual graph of a Hamiltonian outerplanar graph is a tree~(\cite{CH67}). When working with outerplanar graphs we use the term \emph{face} to mean a bounded face.

\subsection{$\ept$ Graphs} \label{subsec:prelim-eptn-knownresults}
We now present definitions and results from the work of \cite{GJ85} that we use throughout this work.

A \emph{pie} of a representation $\rep$ of an $\ept$ graph is an induced star $K_{1,k}$ of $T$ with $k$ leaves $v_0, v_1, \ldots, v_{k-1} \in V(T)$, and $k$ paths $P_0, P_1, \ldots P_{k-1} \in \pp$, such that for every $0 \leq i \leq k-1$ both $v_i$ and $v_{(i+1)\mod k}$ are vertices of $P_i$. We term the central vertex of the star as the \emph{center} of the pie (See Figure \ref{fig:ept-cycle}). It is easy to see that the $\ept$ graph of a pie with $k$ leaves is the hole $C_k$ on $k$ vertices. Moreover, this is the only possible $\ept$ representation of $C_k$ when $k \geq 4$.

\begin{figure}
\centering
\includegraphics{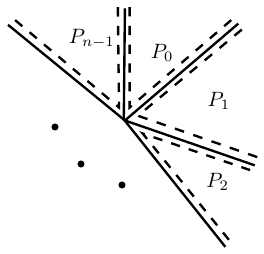}
\caption{The $\ept$ representation of cycle graph on $n$ vertices: a pie.}
\label{fig:ept-cycle}
\end{figure}

\begin{theorem}\cite{GJ85}\label{thm:golumbicpie}
If an $\ept$ graph contains a hole with $k \geq 4$ vertices, then every representation of it contains a pie with $k$ paths.
\end{theorem}

Let $\pp_e \defined \set{p \in \pp|~e \in p}$ be the set of paths in $\pp$ containing the edge $e$. A star $K_{1,3}$ is termed a \emph{claw}. For a claw $K$ of a tree $T$, $\pp[K] \defined \set{p \in \pp|~p \textrm{~uses two edges of~}K}$. It is easy to see that both $\eptg{\pp_e}$ and $\eptg{\pp[K]}$ are cliques. These cliques are termed \emph{edge-clique} and \emph{claw-clique}, respectively. Moreover, these are the only possible representations of cliques.
\begin{theorem}\cite{GJ85}\label{thm:golumbiccliques}
Any maximal clique of an $\ept$ graph with representation $\rep$ corresponds to a subcollection $\pp_e$ of paths for some edge $e$ of $T$, or to a subcollection $\pp[K]$ of paths for some claw $K$ of $T$.
\end{theorem}
Note that a claw-clique is a pie with $3$ leaves.

\subsection{$\enpt$ Graphs and $\ept,\enpt$ Graph Pairs} \label{subsec:eptn-OurPreviousResults}
In this section we present definitions and results from \cite{BESZ13-ENPT1-DAM} that we use throughout this work, introduce new terms, and prove basic results.

\runningtitle{Equivalent and minimal representations:}
We say that the representations $\repn{1}$ and $\repn{2}$ are \emph{equivalent}, and denote by $\repn{1} \approxeq \repn{2}$, if their corresponding $\ept$ and $\enpt$ graphs are isomorphic under the same isomorphism (in other words, if they constitute representations of the same pair of graphs $(G,G')$).

We write $\repn{1} \gtrsim \repn{2}$ or equivalently $\repn{2} \lesssim \repn{1}$ if $\repn{2}$ can be obtained from $\repn{1}$ by zero or more of the following
two operations that we term as \emph{minifying} operations.

\begin{itemize}
\item{} Contraction of an edge $e$ of $T_1$ (and of all the paths in $\pp_1$ using $e$). We denote this operation as $contract(e)$.
\item{} Removal of an initial edge (\emph{tail}) $e$ of a path $P$ of $\pp_1$. We denote this operation as $tr(P,e)$.
\end{itemize}

We say that $\rep$ is a \emph{minimal} representation, if it is minimal in the partial order $\lesssim$ defined over all the representations representing the same pair as $\rep$. Throughout the work we aim at characterizing minimal representations. Whenever a sequence of minifying operations contains two operations $tr(P,e)$ and $contract(e)$ in this order, the first operation can be deleted from the sequence, to obtain a shorter and equivalent sequence. A sequence that can not be shortened in this way and contains the biggest number of $contract(e)$ operations is termed a \emph{minimal} sequence. Whenever $\repn{1} \gtrsim \repn{2}$ we consider only minimal sequences of minifying operations that transform $\repn{1}$ to $\repn{2}$. We observe two properties of such sequences.
\begin{lemma}\label{lem:MinOperationsExchangeable}
Let $\repn{1} \gtrsim \repn{2}$, and $s$ be a minimal sequence of minifying operations transforming $\repn{1}$ to $\repn{2}$.
Then every permutation of $s$ also transforms $\repn{1}$ to $\repn{2}$.
\end{lemma}
\begin{proof}
If $contract(e)$ is an operation of $s$ then there is no other operation in $s$ involving $e$. This is because such an operation is impossible after $contract(e)$, and if it appears before $contract(e)$ it contradicts the minimality of $s$. To conclude the result, we observe that any two successive operations in $s$ are interchangeable. Indeed, for two distinct edges $e,e'$ the operations $contract(e), contract(e')$ (resp. $contract(e),tr(P,e')$) are interchangeable, and for two not necessarily distinct edges $e,e'$ the operations $tr(P,e),tr(P',e')$ are interchangeable.
\end{proof}

\begin{lemma}\label{lem:MinOperationsMonotone}
If $\repn{1} \gtrsim \cdots \gtrsim \repn{n}$ and $\repn{1} \approxeq \repn{n}$, then $\repn{1} \approxeq \cdots \approxeq \repn{n}$.
\end{lemma}
\begin{proof}
Let $G_i=\eptg{\pp_i}$ and $G'_i=\enptg{\pp_i}$. We observe that both minifying operations are monotonic in the sense that they neither introduce neither new intersections, nor new splits. Namely, for $1 \leq i < n$, $E(G_{i+1}) \subseteq E(G_i)$ and $E(G_{i+1}) \setminus E(G'_{i+1}) \subseteq E(G_i) \setminus E(G'_i)$. As $\repn{1} \approxeq \repn{n}$ we have $(G_1,G'_1)=(G_n,G'_n)$, i.e. $E(G_n)=E(G_1)$ and $E(G_n) \setminus E(G'_n) = E(G_1) \setminus E(G'_1)$. Therefore, $E(G_1) = \cdots = E(G_n)$ and $E(G_1) \setminus E(G'_1) = \cdots =  E(G_n) \setminus E(G'_n)$, concluding $(G_1,G'_1) = \cdots =(G_n,G'_n)$.
\end{proof}

\runningtitle{$\ept$ holes:}
\begin{lemma}\label{lem:NoBlueHole}
A hole of size at least $4$ of an $\ept$ graph does not contain blue (i.e. $\enpt$) edges.
\end{lemma}
\begin{proof}
Consider the pie representation of some hole of an $\ept$ graph. For any two paths $P_p,P_q$ of this pie, we have either $P_p \nsim P_q$ or $P_p \parallel P_q$, therefore $\set{p,q}$ is not an
$\enpt$ edge.
\end{proof}

Combining Lemma \ref{lem:NoBlueHole} with Theorem \ref{thm:golumbicpie}, we obtain the following characterization of pairs $(C_k, G')$. For a pair $(C_k$, $G'$), exactly one of the following holds:
\begin{itemize}
\item{$k>3$.} In this case $C_k$ is represented by a pie. Therefore, $G'$ is an independent set. In other words, $C_k$ consists of red edges. We term such a cycle, a red hole.
\item{$k=3$ and $C_k$ consists of red edges.} $G'$ is an independent set. We term such a cycle a red triangle.
\item{$k=3$ and $C_k$ contains exactly one $\enpt$ (blue) edge.} We term such a cycle a $BRR$ triangle, and its representation is an edge-clique.
\item{$k=3$ and $C_k$ contains two $\enpt$ (blue) edges.} We term such a cycle a $BBR$ triangle, and its representation is an edge-clique.
\item{$k=3$ and $C_k$ consists of blue edges ($G' = C_3$).} We term such a cycle a blue triangle, and its representation is an edge-clique.
\end{itemize}

\runningtitle{$\ept$ contraction:}
Let $\rep$ be a representation and $P_p,P_q \in \pp$ such that $P_p \sim P_q$. We denote by $\contractppq$ the representation that is obtained from $\rep$ by replacing the two paths $P_p,P_q$ by the path $P_p \cup P_q$, i.e. $\contractppq \defined \left<T,\pp \setminus \set{P_p,P_q} \cup \set{P_p \cup P_q}\right>$. We term this operation a \emph{union}. Lemma \ref{lem:contraction-union} follows from the below observation.
\begin{observation}\label{obs:splitofunion}
For every $P_p,P_q,P_r \in \pp$ such that $P_p \sim P_q$, $\split(P_p \cup P_q, P_r) = \split(P_p,P_r) \cup \split(P_q,P_r)$.
\end{observation}

\begin{lemma}\label{lem:contraction-union}
\cite{BESZ13-ENPT1-DAM} Let $\rep$ be a representation for the pair $(G,G')$, and let $e=\set{p,q} \in E(G')$. Then
$\contractge$ is an $\ept$ graph. Moreover, $\contractge=\eptg{\contractppq}$.
\end{lemma}

\runningtitle{Contraction of pairs:}
The definition of the contraction operation extends to pairs: The contraction of an $\enpt$ edge does not necessarily correspond to the union operation in the $\enpt$ representation. For example, let $P_p$,$P_q$ and $P_{q'}$ be such that $P_p \sim P_q$, $P_p \sim P_{q'}$ and $P_q \nsim P_{q'}$. Then $\contract{G'}{\set{p,q}}$ is not isomorphic to $\enptg{\contractppq}$ as $ \set{q',p.q} \notin E(\enptg{\contractppq})$. Let $(G,G')$ be a pair and $e \in E(G')$. If for every edge $e' \in E(G')$ incident to $e$, the edge $e''= e \triangle e'$ (forming a triangle together with $e$ and $e'$) is not an edge of $G$ then $\contract{(G,G')}{e} \defined (\contractge,\contract{G'}{e})$, otherwise $\contract{(G,G')}{e}$ is undefined. Whenever $\contract{(G,G')}{e}$ is defined we say that $(G,G')$ is \emph{contractable} on $e$, or when there is no ambiguity about the pair under consideration we say that $e$ is \emph{contractable}. We say that a pair is contractable if it is contractable on some $\enpt$ edge. Clearly, $(G,G')$ is not contractable if and only if every edge of $G'$ is contained in at least one $BBR$ triangle.

\runningtitle{Weak Dual Trees:}
The definition of weak dual tree is extended from Hamiltonian outerplanar graphs to any Hamiltonian graph as follows. Given a pair $(G,C)$ where $C$ is a Hamiltonian cycle of $G$, a weak dual tree of $(G,C)$ is the weak dual tree $\wdtg$ of an arbitrary Hamiltonian maximal outerplanar subgraph $\opg$ of $G$. $\opg$ can be built by starting from $C$ and adding to it arbitrarily chosen chords from $G$ as long as such chords exist and the resulting graph is planar.

By definition of a dual graph, vertices of $\wdtg$ correspond to faces of $\opg$. By maximality, the faces of $\opg$ correspond to holes of $G$. The degree of a vertex of $\wdtg$ is the number of red edges in the corresponding face of $\opg$. To emphasize the difference, for an outerplanar graph $G$ we will say \emph{the} weak dual tree of $G$, whereas for a (not necessarily outerplanar) graph $G$ we will say \emph{a} weak dual tree of $G$.

The following lemma describes the effect of contraction on weak dual trees.
\begin{lemma}\label{lem:contraction-wdt}
\cite{BESZ13-ENPT1-DAM} Let $(G,C)$ be a pair satisfying $(P2),(P3)$ and let $\wdtg$ be a weak dual tree of $(G,C)$. (i) There is a bijection between the contractable edges of $(G,C)$ and the intermediate vertices of $\wdtg$. (ii) The tree obtained from $\wdtg$ by smoothing out the intermediate vertex corresponding to a contractable edge $e$ is a possible weak dual tree of $\contractgce$.
\end{lemma}

\runningtitle{Representations of $\enpt$ holes:}
Our goal in this work is to characterize the representations of $\enpt$ holes. More precisely we characterize representations of pairs $(G,C_n)$ where $C_n$ is a Hamiltonian cycle of $G$. For this purpose we define the following problem.

\begin{center} \fbox{\begin{minipage}{11.9cm}
\noindent  \prb

{\bf Input:} A pair $(G,C_n)$ where $C_n$ is a Hamiltonian cycle of $G$

{\bf Output:} A minimal representation $\rep$ of $(G,C_n)$ if such a representation exists, ``NO'' otherwise.

\end{minipage}}\\
\end{center}

For $n=3$ the only possible pair is $(C_3,C_3)$ whose unique minimal representation is by 3 identical paths consisting of one edge each.

Let $T$ be a tree with $k$ leaves and $\pi=(\pi_0,\ldots,\pi_{k-1})$ a cyclic permutation of the leaves. The \emph{tour} $(T,\pi)$ is the following set of $2k$ paths: $(T,\pi)$ contains $k$ \emph{long} paths, each of which connecting two consecutive leaves $\pi_i, \pi_{i+1 \mod k}$. $(T,\pi)$ contains $k$ \emph{short} paths, each of which connecting a leaf $\pi_i$ and its unique neighbor in $T$.

A \emph{planar embedding} of a tour is a planar embedding of the underlying tree such that any two paths of the tour do not cross each other. A tour is \emph{planar} if it has a planar embedding. Note that a tour $(T,\pi)$ is planar if and only if $\pi$ corresponds to the order in which the leaves are encountered by some DFS traversal of $T$.

Consider the following three properties
\begin{itemize}
\item{$(P1)$:} $(G,C_n)$ is not contractable.
\item{$(P2)$:} $(G,C_n)$ is $(K_4, P_4)$-free, i.e., it does not contain an induced sub-pair isomorphic to a $(K_4, P_4)$.
\item{$(P3)$:} Every red triangle of $(G,C_n)$ is a claw-clique, i.e. corresponds to a pie of $\rep$.
\end{itemize}

Note that $(P1)$ and $(P2)$ are properties of pairs and $(P3)$ is a property of representations. We say that $(P3)$ holds for a pair $(G,C)$ whenever it has a representation $\rep$ satisfying $(P3)$. It is convenient to define the following problem.

\begin{center} \fbox{\begin{minipage}{11.9cm}
\noindent  \prbpthree

{\bf Input:} A pair $(G,C_n)$ where $C_n$ is a Hamiltonian cycle of $G$ and $n \geq 4$.

{\bf Output:} A minimal representation $\rep$ of $(G,C_n)$ that satisfies $(P3)$ if such a representation exists, ``NO '' otherwise.

\end{minipage}}\\
\end{center}

In this work we extend the following results of \cite{BESZ13-ENPT1-DAM}.
\begin{theorem}\label{thm:BuildPlanarTourCorrect}
\cite{BESZ13-ENPT1-DAM} Instances of $\prbpthree$ satisfying properties $(P1),(P2)$ can be solved in polynomial time. YES instances have a unique solution, and whenever $n \geq 5$ this solution is a planar tour.
\end{theorem}

\begin{theorem}\label{theo:P123}
\cite{BESZ13-ENPT1-DAM} If $n>4$ the following statements are equivalent:\\
(i) $(G,C_n)$ satisfies assumptions $(P1-3)$.\\
(ii) $(G,C_n)$ has a unique minimal representation satisfying $(P3)$ which is a planar tour of a weak dual tree of $G$.\\
(iii) $G$ is Hamiltonian outerplanar and every face adjacent to the unbounded face $F$ is a triangle having two edges in common with $F$, (i.e. a BBR triangle).
\end{theorem}

For $n=4$ there are two possible pairs, namely $\kp$ and $(K_4 - e, C_4)$, each of which satisfies $(P1),(P2)$ and has a unique minimal representation. Therefore, in this work we assume $n \geq 5$.

The opposite of a sequence of union operations that create one path is termed \emph{breaking apart}. Namely, breaking apart a path $P$ is to replace it with paths $P_1, \ldots ,P_k$ such that $\cup_{i=1}^k P_i=P$, $\forall 1 \leq i < k, P_i \cap P_{i+1} \neq \emptyset$, and $P_i \subseteq P_j$ if and only if $i=j$.
A \emph{broken tour} is a representation obtained from a tour by subdividing edges and breaking apart long paths of the tour. Clearly, if the tour is planar the broken tour is also planar, i.e. has a planar embedding.

\section{Pairs $(G,C)$ Satisfying $(P2)$ and $(P3)$}\label{sec:uncontraction}
In this section our goal is to get rid of the assumption that $(P1)$ holds, so that to characterize representation of pairs that satisfy only $(P2)$ and $(P3)$. Recall that property $(P1)$ states that the pair is uncontractable. We therefore, consider contractable pairs and examine the effects of the contraction on the representation. In this way we reduce the problem to the base case of uncontractable pairs where all assumptions hold for which the representations are already known. In Section \ref{subsec:ContactionOfPairs} we investigate the basic properties of the contraction operation. In Section \ref{subsec:c5c6} we investigate the case of small cycles (i.e., those of size at most 6) that require special handling due to the nature of the base case (recall that Theorems \ref{theo:P123} and \ref{thm:BuildPlanarTourCorrect} do not hold for a cycle of length 4). Finally, in Section \ref{subsec:P23Representations} we develop an algorithm for the general case, i.e. cycles bigger than 6.
\subsection{Contraction of Pairs}\label{subsec:ContactionOfPairs}
In this section we investigate properties of the contraction operation in our goal to extend Theorem \ref{theo:P123} to contractable pairs. More specifically, we characterize representations of pairs $(G,C)$ satisfying $(P2),(P3)$. We show that (i) the contraction operation preserves $\enpt$ edges, (ii) the order of contractions is irrelevant and (iii) the contraction operation preserves properties $(P2)$ and $(P3)$.

\begin{lemma}\label{lem:contraction-pairs}
Let $\rep$ be a representation for the pair $(G,G')$, and let $e=\set{p,q} \in E(G')$. If $\contract{(G,G')}{e}$ is defined then
$\contractppq$ is a representation for the pair $\contract{(G,G')}{e}$.
\end{lemma}
\begin{proof}
By Lemma \ref{lem:contraction-union} $\contractppq$ is an $\ept$ representation for $\contractge$. It remains to show that it is an $\enpt$ representation for $\contract{G'}{e}$, i.e. that for any two paths $P_{p'},P_{q'} \in \contractppq$, the edge $e'=\set{p',q'}$ is in $E(\contract{G'}{e}) \iff P_{p'} \sim P_{q'}$. Let $P_s = P_p \cup P_q$ and $s$ be the vertex obtained by the contraction. We assume first that $P_s \notin \set {P_{p'},P_{q'}}$. Then $e' \in E(\contract{G'}{e}) \iff e' \in E(G') \iff P_{p'} \sim P_{q'}$ as required. Now we assume without loss of generality that $P_{p'}=P_s$ and we recall that $e=\set{p,q} \in E(G')$ is the contracted edge. We have to show that $e'=\set{s,q'} \in E(\contract{G'}{e}) \iff P_s \sim P_{q'}$. We observe that
\begin{eqnarray}
\set{s,q'} \in E(\contract{G'}{e}) & \iff & \set{p,q'} \in E(G') \vee \set{q,q'} \in E(G') \nonumber\\
& \iff &  P_p \sim P_{q'} \vee P_q \sim P_{q'}\label{eqn:one}
\end{eqnarray}
and
\begin{eqnarray}
P_s \sim P_{q'} & \iff & (P_p \cup P_q) \cap P_{q'} \neq \emptyset \wedge \split(P_p \cup P_q, P_{q'})=\emptyset \nonumber\\
& \iff & (P_p \cap P_{q'} \neq \emptyset \vee P_q \cap P_{q'} \neq \emptyset) \wedge \split(P_p, P_{q'})=\emptyset \wedge \split(P_q, P_{q'})=\emptyset.\label{eqn:two}
\end{eqnarray}
Clearly, (\ref{eqn:two}) implies (\ref{eqn:one}). To conclude the proof, assume that ($\ref{eqn:one}$) holds. Then $P_p \cap P_{q'} \neq \emptyset \vee P_q \cap P_{q'} \neq \emptyset$. Now assume, by way of contradiction, that (\ref{eqn:two}) does not hold. Then $\split(P_p,P_{q'}) \neq \emptyset \vee \split(P_q,P_{q'}) \neq \emptyset$ implying $P_p \nsim P_{q'} \vee P_q \nsim P_{q'}$. Combining with (\ref{eqn:one}) this implies that exactly one of $P_p \sim P_{q'}$ and $P_q \sim P_{q'}$ holds. Therefore, without loss of generality $P_p \sim P_{q'}, P_q \nsim P_{q'}$. Then $e'=\set{p,q'} \in E(G')$ and $e \triangle e' \in E(G)$, therefore $\contract{(G,G')}{e}$ is undefined, thus constituting a contradiction to the assumption of the lemma.
\end{proof}

Let $\bar{E}=\set{e_1, e_2, \ldots, e_k} \subseteq E(G')$. For every $k>1$ we define $\contract{(G,G')}{e_1,\ldots,e_k} \defined \contract{\contract{(G,G')}{e_1,\ldots,e_{k-1}}}{e_k}$  provided that both contractions on the right hand side are defined, otherwise it is undefined. The following Lemma follows from Lemma \ref{lem:contraction-pairs} and states that the order of contraction of the edges is irrelevant.

\begin{lemma}\label{lem:orderIrrelevantInContraction}
Let $(G,G')$ be a pair, $\bar E= \set{e_1, e_2, \ldots, e_k} \subseteq E(G)$, and $\pi$ a permutation of the integers $\set{1,\ldots, k}$. Then $\contract{(G,G')}{e_1,\ldots,e_k}$ is defined if and only if $\contract{(G,G')}{e_{\pi(1)},\ldots,e_{\pi(k)}}$ is defined. Moreover, whenever they are defined, we have $\contract{(G,G')}{e_1,\ldots,e_k}=\contract{(G,G')}{e_{\pi(1)},\ldots,e_{\pi(k)}}$.
\end{lemma}
\begin{proof}
Assume that $\contract{(G,G')}{e_1,\ldots,e_k}$ is defined. By $k-1$ successive applications of Lemma \ref{lem:contraction-pairs} we conclude that a representation of $\contract{(G,G')}{e_1,\ldots,e_k}$ can be obtained from a representation of $(G,G')$ by applying a sequence of $k-1$ union operations. The result of $k-1$ union operations is a set of paths. As the union operation is commutative and associative, this set of paths is independent of the order of the union operations. On the other hand as union preserves split vertices, and the result does not contain split vertices, there are no split vertices at any given step of new sequence of union operations. We conclude that $\contract{(G,G')}{e_{\pi(1)},\ldots,e_{\pi(k)}}$ is defined. The other direction holds by symmetry. Whenever both $\contract{(G,G')}{e_1,\ldots,e_k}$ and $\contract{(G,G')}{e_{\pi(1)},\ldots,e_{\pi(k)}}$ are defined we have
$\contract{(G,G')}{e_1,\ldots,e_k}$=$(\contract{G}{e_1,\ldots,e_k},\contract{G'}{e_1,\ldots,e_k})$=
$(\contract{G}{e_\pi(1),\ldots,e_\pi(k)},\contract{G'}{e_\pi(1),\ldots,e_\pi(k)})=\contract{(G,G')}{e_{\pi(1)},\ldots,e_{\pi(k)}}$.
\end{proof}

Based on this result, we denote the contracted pair as $\contract{(G,G')}{\bar{E}}$ and say that $\bar{E}$ is \emph{contractable}.

We write $\repn{1} \gtrsim_U \repn{2}$, or equivalently $\repn{2} \lesssim_U \repn{1}$ the fact that the representation $\repn{2}$ can be obtained from $\repn{1}$ by applying zero or more union operations. Similarly, we write $(G_1,G'_1) \gtrsim_C (G_2, G'_2)$ if the pair $(G_2, G'_2)$ can be obtained from $(G_1, G'_1)$ by applying zero or more (valid) edge contraction operations. 
By Lemma \ref{lem:contraction-pairs}, $\lesssim_U$ is homomorphic to $\lesssim_C$.
Following the above definitions, a non-contractable pair of graphs is said to be \emph{contraction-minimal}, because it is minimal in the partial order $\lesssim_C$.

We proceed by showing that the contraction operation preserves assumptions $(P2),(P3)$.

\begin{lemma}\label{lem:BBRTopDoesNotComeFromContraction}
Let $\set{p,q,r}$ be a $BBR$ triangle of $\contract{(G,G')}{e}$ where $\set{p,r}$ is the red edge. Then $q \in V(G)$, i.e. $q$ is not the vertex obtained by the contraction.
\end{lemma}
\begin{proof}
Assume, by contradiction that $e=\set{q',q''}$ and $q$ is the vertex obtained by the contraction of $e$. Assume without loss of generality that $\set{p,q'}$ and $\set{q'',r}$ are edges of $G'$. Then both $\set{p,q''}$ and $\set{r,q'}$ are non-edges of $G$, because otherwise $e$ is not contractable. Then $\set{p,q',q'',r}$ is a hole of size $4$ with blues edges, a contradiction.
\end{proof}

\begin{lemma}\label{lem:ContractionPreservesP2P3}
(i) If $(P2)$ holds for $(G,G')$ then $(P2)$ holds for $\contract{(G,G')}{e}$.\\
(ii) If $(P3)$ holds for $(G,G')$ then $(P3)$ holds for $\contract{(G,G')}{e}$.\\
\end{lemma}
\begin{proof}
(i) Assume, by contradiction, that $(G,G')$ does not have an induced sub-pair isomorphic to $\kp$ and without loss of generality $\contract{(G,G')}{e}$ has a sub-pair isomorphic to $\kp$ induced by the vertices $U=\set{p,q,r,s}$ where $p$ and $s$ are the endpoints of the subgraph isomorphic to $P_4$. Let $v$ be the vertex created by the contraction of $e$. If $v \notin U$ then the sub-pair induced by $U$ is also a sub-pair of $(G,G')$, contradicting our assumption. Therefore, $v \in U$. By Lemma \ref{lem:BBRTopDoesNotComeFromContraction} we have that $v \notin \set{q,r}$. Therefore, let without loss of generality $v=p$, $e=\set{p',p''}$ and $p''$ is adjacent to $q$ in $G'$. $\set{p',q}$ is a non-edge of $G$, because $e$ is contractable. As $\set{p,s}$ and $\set{p,r}$ are edges of $\contractge$, $\set{p',s}$ or $\set{p'',s}$ is an edge of $G$, and $\set{p',r}$ or $\set{p'',r}$ is an edge of $G$. If  $\set{p'',s}$ is an edge of $G$ then  $\set{p'',r}$ is a non-edge of $G$ since otherwise  $\set{p'',q,r,s}$ induce a sub-pair isomorphic to $\kp$ in $(G,G')$. Therefore, $\set{p',r}$  is an edge of $G$. Then  $\set{p',p'',q,r}$ induces a hole of size 4 with blue edges, a contradiction. Thus $\set{p',s}$ is an edge of $G$, since $\set{p',q}$ and $\set{p'',s}$ are non-edges of $G$, $\set{p',p'',q,s}$ induce a hole of length 4 with blue edges, a contradiction.

(ii) Assume, by contradiction, that $(G,G')$ has a representation $\rep$ satisfying $(P3)$ and that no representation of $\contract{(G,G')}{e}$ satisfies $(P3)$. Let $v$ be the vertex created by the contraction of $e=\set{v',v''}$. Then by Lemma \ref{lem:contraction-union} $\contract{\rep}{P_{v'},P_{v''}}$ is a representation of $\contract{(G,G')}{e}$ and it contains a red edge-clique $\set{p,q,r}$. If $v \notin \set{p,q,r}$ then $\set{p,q,r}$ is an edge-clique of $\rep$, contradicting our assumption. Assume without loss of generality that $v=p$. Let $e_0$ be an edge of $T$ defining the edge-clique $\set{v,q,r}$. Since $e_0 \in P_{v'} \cup P_{v''}$, without loss of generality $e_0 \in P_{v'}$. Then $\set{v',q,r}$ induces an edge-clique on $e_0$. This is clearly a red edge-clique since $e$ is contractable, a contradiction.
\end{proof}

We now describe how minimal representations of $\contract{(G,G')}{e}$ can be obtained from minimal representations of $(G,G')$.
\begin{lemma}\label{lem:MinimizationAfterUnion}
Let $\rep$ be a minimal representation, $\repprime$ a representation such that $\repprime \lesssim \contractppq$ and $\repprime  \approxeq \contractppq$. Let $e$ be an edge of $T$ involved in a minimal sequence of minifying operations $s$ that obtains $\repprime$ from $\contractppq$. There is an operation of $s$ and a path $P$ such that the operation removes $e$ from $P$ (i.e., $s$ is either $tr(P,e)$ or $contract(e)$, and $e \in P$) where at least one of the following holds:\\
(i) $e$ is a tail of $P_p \cap P_q$, $P \cap P_p \cap P_q=\set{e}$ and $P \cap (P_p \cup P_q) \supsetneq \set{e}$.\\
(ii) $e$ is incident to an internal vertex $u$ of $P_p \cup P_q$, $e$ is a tail of $P$, $P$ is not in a pie with center $u$.
\end{lemma}
\begin{proof}
Let $G=\eptgp$ and $G'=\enptgp$ and consider an operation $op$ of $s$. Without loss of generality we can assume that $op$ is the first operation of $s$, by Lemma \ref{lem:MinOperationsExchangeable}. Furthermore, by Lemma \ref{lem:MinOperationsMonotone}, the representation obtained by applying $op$ is also equivalent to $\contractppq$. Therefore, without loss of generality $op$ is the only operation of $s$. Note that $op$ is defined on $\contractppq$ except when $op$ is $tr(P_p \cup P_q,e)$. In this case $e$ is a tail of $P_p$ or $P_q$ (or both). In the following discussion, whenever we apply $op$ to $\rep$, we mean that we apply $tr(P_p,e)$ or $tr(P_q,e)$ one of which is well defined on $\rep$.

By the minimality of $\rep$, $op$ cannot be applied to $\rep$. More precisely, if $op$ is applied, either an edge of $G$  becomes a non-edge, or a red edge of $(G,G')$ becomes a blue edge. We term such an edge of $(G,G')$ \emph{an affected edge} and the corresponding paths of $\rep$ \emph{affected pair of paths}. Let $\set{r,s}$ be an affected edge of $(G,G')$. If $\set{P_r,P_s} \cap \set{P_p,P_q} = \emptyset$ then $P_r,P_s$ is a pair of affected paths in $\contractppq$, contradicting the fact that $op$ can be applied to $\contractppq$. We conclude that $\set{P_r,P_s} \cap \set{P_p, P_q} \neq \emptyset$. Assume without loss of generality that $P_s=P_p$, i.e. $P_r,P_p$ is an affected pair of paths. We consider two disjoint cases:

Case 1) $\set{r,p}$ becomes a non-edge after applying $op$. Then $P_r \cap P_p = \set{e}$ for some edge $e$ of $T$, and after the removal of $e$ the intersection becomes empty. On the other hand $P_r \cap (P_p \cup P_q) \supsetneq \set{e}$, because otherwise $P_r$ and $P_p \cup P_q$ constitute an affected pair of paths in $\contractppq$. Then $e$ is a tail of $P_p \cap P_q$ (see Figure \ref{fig:MinimizationAfterUnion} a) and $P_r$ is the claimed path $P$ (note that possibly $r=q$ as opposed to the figure). In this case i) holds.

Case 2) $\set{r,p}$ is a red edge, and it becomes a blue edge after applying $op$. Then $P_r \nsim P_p$ (therefore $r \neq q$) and $P_r,P_p$ do not split after applying $op$. Therefore, $\split(P_r,P_p)=\set{u}$ for an endpoint $u$ of $e$, $e$ is a tail of exactly one of $P_r, P_p$, and $u$ is an internal vertex of $P_p$ thus of $P_p \cup P_q$. Let $\set{P,P'} = \set{P_r, P_p}$ such that $e$ is a tail of $P$, and $e \notin P'$. If $P$ is not in a pie with center $u$ then ii) holds. Otherwise $P$ has two neighbors $\set{P',P''}$ in this pie. $e \in P''$ because $e \notin P$ and $e$ is an edge incident to $u$, the center of the pie. Recalling that $e$ is a tail of $P$ we conclude that after the removal of $e$ from $P$, its intersection with $P''$ becomes empty. Therefore, i) holds.
\end{proof}
\begin{figure}[htbp]
\centering
\includegraphics{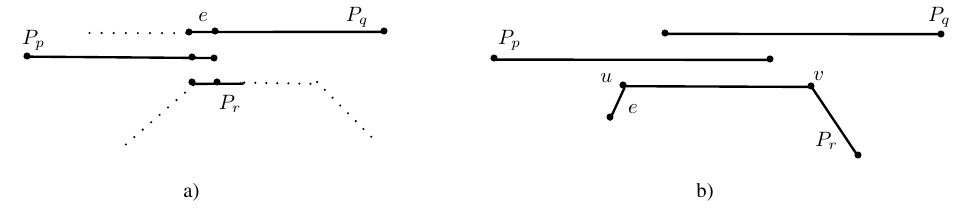}
\caption{Possible minifying operations on $\contractppq$.}
\label{fig:MinimizationAfterUnion}
\end{figure}

\begin{lemma}\label{lem:EverySplitVertexIsACenterofBrokenPlanarTour}
Every split vertex of a path $P$ of a broken planar tour is a center of a pie containing $P$.
\end{lemma}
\begin{proof}
By construction, every split vertex of a path $P$ of a tour is a center of a pie containing $P$. We will show that the same holds for a broken planar tour. Let $P_1',P_2'$ be two paths of a broken planar tour such that $v \in \split(P_1',P_2')$. These paths are sub-paths of two paths $P_1,P_2$ of a tour and $v \in \split(P_1,P_2)$. Then $v$ is a center of a pie containing $P_1,P_2$ and also other paths. Each one of the other paths has at least one sub-path in the broken planar tour that crosses $v$. These paths, together with $P_1',P_2'$ constitute a pie with center $v$ of the broken planar tour. We conclude that every split vertex of $P$ is a center of a pie.
\end{proof}

We notice that by Lemma \ref{lem:EverySplitVertexIsACenterofBrokenPlanarTour} it follows that the case (ii) of Lemma \ref{lem:MinimizationAfterUnion} is impossible.

\newcommand{\algbpt}{\textsc{BuildPlanarTour}}
\subsection{Small Cycles: the pairs $(G,C_5)$ and $(G,C_6)$}\label{subsec:c5c6}
We now return to the study of the representations of pairs $(G',C')$ satisfying $(P2),(P3)$. Without loss of generality we let $V(G')=V(C')=\set{0,1,\ldots,n-1}$, $n \geq 5$ and note that all arithmetic operations on vertex numbers are done modulo $n$.

In this Section we analyze the special cases of $n=5$ and  $n=6$. These cases are special because our technique for the general case is based on contraction of cycles to smaller ones and assumes that the representation of a non-contractable pair is a planar tour (Theorem \ref{theo:P123}). However this theorem does not hold when $n=4$. The following lemma analyzes the case $n=5$. We note that in this case $(P3)$ holds vacuously, since it can be easily seen that such a pair does not contain a red triangle. In Lemma \ref{lem:NoContractableC6} we analyze the case $n=6$ for which we show that whenever $(P2)$ and $(P3)$ hold, $(P1)$ holds too, implying that the only representations for this case are those implied by Theorem \ref{thm:BuildPlanarTourCorrect}.

\begin{lemma}\label{lem:unique-c5}
If $(G',C_5)$ satisfies $(P2)$ then (i) $G'$ is the graph depicted in Figure \ref{fig:eptn-c5}, and (ii) $(G',C_5)$ has a unique minimal representation also depicted in Figure \ref{fig:eptn-c5}.
\end{lemma}
\begin{proof}
(i) $G'$ contains at least two non-crossing red edges, because otherwise there is a hole of size $4$ with blue edges. Without loss of generality, let these edges be $\set{1,3}$ and $\set{1,4}$. If one of $\set{2,4}$ or $\set{0,3}$ is a red edge, then we have a $\kp$, contradicting our assumption. If $\set{0,2}$ is a red edge, then we have a hole of size $4$ containing blue edges, contradicting Lemma \ref{lem:NoBlueHole}. Therefore, $\set{1,3}$ and $\set{1,4}$ are  the only red edges in this pair.

(ii) We contract $\set{3,4}$ of $(G',C_5)$ and obtain the pair $(G,C_4)$ with one red edge $\set{1,3.4}$. This pair has a unique minimal representation $\repprime$ characterized in \cite{BESZ13-ENPT1-DAM}. Any representation of $(G', C_5)$ is obtained by splitting the path $P'_{3.4}$ of $\repprime$ into two overlapping paths and making sure that both of them split from $P_1$. This leads to the minimal representation depicted in Figure \ref{fig:eptn-c5}.
\end{proof}

\begin{figure}[htbp]
\centering
\includegraphics{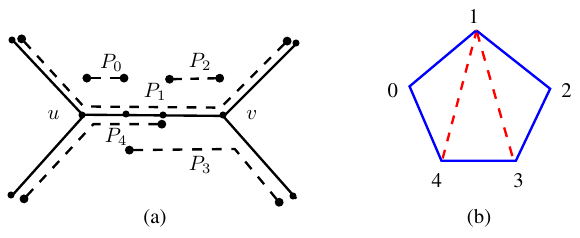}
\caption{(a) The unique $\enpt$ representation of $C_5$ satisfying $(P2)$ and (b) corresponding pair $(G,C_5)$.}
\label{fig:eptn-c5}
\end{figure}

\begin{lemma}\label{lem:NoContractableC6}
If $(G',C_6)$ satisfies $(P2)$ and $(P3)$ then it is not contractable, i.e. it satisfies $(P1)$ too.
\end{lemma}
\begin{proof}
Assume, by way of contradiction, that $(G',C_6)$ satisfies $(P2)$ $(P3)$ and the edge $e=\set{0,1}$ is contractable. Therefore, $\set{0,2}$ and $\set{5,1}$ are non-edges of $G'$. $\set{2,5}$ is also a non-edge, because otherwise $\set{0,1,2,5}$ is a hole of size $4$ with blue edges. Then $\set{0,1}$ must be in a $BRR$ triangle. From the two possible options remaining, assume without loss of generality that this triangle is $\set{0,1,4}$. At least one of $\set{2,4}$ and $\set{1,3}$ is an edge of $G'$ because otherwise $\set{1,2,3,4}$ is a hole of size $4$ with blue edges. On the other hand, if both of them are edges then $\set{1,2,3,4}$ is a $\kp$, a contradiction. Therefore, exactly one of them is an edge of $G'$. We analyze these cases separately.
\begin{itemize}

\item{$\set{2,4}$ is an edge of $G'$, $\set{1,3}$ is not an edge of $G'$:} In this case $\set{0,3}$ is not an edge, because otherwise $\set{0,1,2,3}$ is a hole of size $4$ with blue edges. Then $\set{3,5}$ is not an edge, because otherwise $\set{0,1,2,3,5}$ is a hole of size $5$ with blue edges. Then $\set{5,0,1,2,3}$ induces a path on $4$ vertices in $G'$. Since none of the paths $P_0,P_1,P_2,P_3,P_5$ split from another, their union is a graph with maximum degree two, i.e. every representation of them is an interval representation where no three paths intersect at one edge. Now $P_4 \sim P_5$ and $P_4 \sim P_3$. Therefore, $P_4$ intersects all of $P_0,P_1$ and $P_2$ and does not split from them. Then $\set{4,0},\set{4,1},\set{4,2}$ are blue edges, a contradiction.

\item{$\set{1,3}$ is an edge of $G'$, $\set{2,4}$ is not an edge of $G'$:} Assume by way of contradiction $\set{0,1}$ is contracted, the contracted pair is the same as the pair in Figure \ref{fig:eptn-c5} (b) where contracted edge $\set{0,1}$ corresponds to vertex $1$ of $(G,C_5)$. We will show that $1$ can not be a vertex obtained by a contraction. Let $\set{1',1''}$ be the contracted edge. For the following discussion consult Figure \ref{fig:eptn-c5} (a). One endpoint of each one of $P_{1'}$, $P_{1''}$ is the same as the endpoints of $P_1$ since $P_1 = P_{1'} \cup P_{1''}$. $P_{1'}$ (resp. $P_{1''}$) can not cross $v$ since otherwise $\set{1',2}$ (resp. $\set{0,1''}$) is a blue chord. $P_{1'} \sim P_{1''}$, therefore there exist some edge $e$ such that $P_{1'} \cap P_{1''} \ni e$ and $e \in p_T(u,v)$. But $e \in p_T(u,v) \subseteq P_3 \cup P_4$ then either $\set{1',3}$ or $\set{1'',4}$ (or both) is a blue chord.

\end{itemize}
\end{proof}

\subsection{The General Case}\label{subsec:P23Representations}
\newcommand{\algPTwoPThree}{\textsc{FindMinimalRepresentation-P2-P3}}
\newcommand{\procAdjustEndpoint}{\textsc{AdjustEndpoint}}

Algorithm \ref{alg:findminimalrepresentation} is a recursive description of $\algPTwoPThree$. It follows the paradigm of obtaining a non-contractable pair by successive contractions, and then reversing the corresponding union operations in the representation. The reversal of the union operation, i.e. the breaking apart of a path is done by a) duplicating the path, then b) moving one endpoint of each path to a properly chosen internal vertex of the original path, and possibly c)subdividing an edge. The key to the correctness of the algorithm is the following lemma that among others enables us to consider only one minifying operation.

\begin{lemma}\label{lem:MinimizationOfBrokenPlanarTours}
Let $\rep$ be a minimal representation of $(G,C)$, $\repprime$ a broken planar tour representation such that $\repprime \lesssim \contractppq$ and $\repprime  \approxeq \contractppq$. Every operation in a minimal sequence of operations that obtains $\repprime$ from $\contractppq$ is a $contract(e)$ operation, where $e$ is a tail of $P_p \cap P_q$.
\end{lemma}

\begin{proof}
Consider an operation in a minimal sequence of minifying operations as in the statement of the lemma. Let $e$ be the edge involved in the operation, and let $P_r$ be a path whose existence is guaranteed by Lemma \ref{lem:MinimizationAfterUnion}. By Lemma \ref{lem:EverySplitVertexIsACenterofBrokenPlanarTour}, case (ii) of Lemma \ref{lem:MinimizationAfterUnion} is impossible. Then case (i) of the lemma holds, i.e. there is a path $P_r$ such that a) the minifying operation removes $e$ from $P_r$, b) $e$ is a tail of $P_p \cap P_q$, c) $P_r \cap P_p \cap P_q = \set{e}$, and d) $P_r \cap (P_p \cup P_q) \supsetneq \set{e}$.

The minifying operation is either $contract(e)$ or $tr(P_r,e)$. We will show that if $tr(P_r,e)$ can be applied, i.e. there is no affected pair after applying $tr(P_r,e)$, then $contract(e)$ can also be applied. For the following discussion consult Figure \ref{fig:MinimizationAfterUnion} a) where $\split(P_r,P_p) = \emptyset$, i.e. the dotted part of $P_r$ adjacent to $e$ in the figure, is empty.

Without loss of generality we assume that $e$ is a tail of $P_p$. Since $e$ is not a tail of $P_p \cup P_q$, we have $r \neq p.q$. $e$ divides $T$ into two subtrees $T_1, T_2$. As $e$ is a tail of $P_p$, $P_p$ can not intersect both subtrees. We assume without loss of generality that $T_2 \cap P_p = \emptyset$. Let $\bar{\pp}$ denote the set of paths of $\contractppq$, i.e. $\bar{\pp}=\pp \setminus \set{P_p, P_q} \cup \set{P_p \cup P_q}$ and $e'$ be the edge adjacent to $e$ in $P_r \cap (P_p \cup P_q)$. Every path of $P \in \bar{\pp}$ that contains $e$ contains also $e'$, because otherwise $P \cap P_r = \set{e}$ and $(P,P_r)$ would constitute an affected pair of $tr(P_r,e)$.  For $k \in \set{1,2}$, let $\pp_k = \set{P \in \bar{\pp}|~ P \cap T_k \neq \emptyset \wedge e \textrm{~is a tail of~} P}$. Note that by definition, $\pp_1 \cap \pp_2 = \emptyset$. As $e' \in T_2 \cap P_r$, we have $P_r \in \pp_2$. We note that for every path $P_s \in \pp_2$, $P_p \sim P_s$, i.e. $\set{p,s}$ is an edge of $C$. As the degree of $p$ is $2$ in $C$ and both of $q$ and $r$ neighbors of $p$ in $C$, we conclude that $\pp_2 = \set{P_r}$. On the other hand, $\pp_1=\emptyset$ because for every path $P_s \in \pp_1$, $\set{s,r}$ is an affected pair of $tr(P_r,e)$ (as $P_s \cap P_r = \set{e}$). Therefore, $\pp_1 \cup \pp_2 = \set{P_r}$, i.e. the only path with tail $e$ is $P_r$.

Assume by way of contradiction that there exists an affected pair $\set{s,t}$ of $contract(e)$. As $e' \in P_s \cap P_t$, they intersect after the contraction. Therefore, $\set{s,t}$ is a red-edge that becomes blue after the contraction. This can happen only if $e$ is a tail of exactly one of $P_s, P_t$. Therefore, $r \in \set{s,t}$ from the above discussion. But then $\set{s,t}$ constitute an affected pair of $tr(P_r,e)$, contradicting to our initial assumption. We conclude that $contract(e)$ has no affected pairs.
\end{proof}

\alglanguage{pseudocode}
\begin{algorithm}
\caption{$\algPTwoPThree(G',C')$}
\label{alg:findminimalrepresentation}
\begin{algorithmic}[1]
\Require {$C'=\set{0,1,\ldots,\abs{V(G')}-1}$ is an Hamiltonian cycle of $G'$ and $\abs{V(G')} > 5$}
\Ensure {A minimal representation $\repbarprime$ of $(G',C')$ satisfying $(P3)$ if any}
\If {$(G',C')$ is contraction-minimal}
\If{$G'$ is outerplanar}
\State \Return $\algbpt(G',C')$
\Else
\State \Return ``NO''
\EndIf
\EndIf
\State
\Statex \textbf{Contract:}
\State Pick an arbitrary contractable edge $e = \set{i,i+1}$ of $C'$
\State $(G,C) \gets \contract{(G',C')}{e}$
\State Let $j$ be the vertex of $(G,C)$ created by the contraction of the edge $e$
\State
\Statex \textbf{Recurse:}
\State $\repbar \gets \algPTwoPThree(G,C)$.\label{lin:recurse}
\State
\Statex \textbf{Uncontract:}
\State $\repbarprime \gets \repbar$
\State Let $u$ and $v$ be the endpoints of $P_j$ such that
\State ~~~~$u$ (resp. $v$) is contained in $P_{i-1}$ (resp. $P_{i+2}$)
\State Replace $P_j \in \bar{\pp}'$ by two copies $P_i$ and $P_{i+1}$ of itself
\State $\procAdjustEndpoint (\repbarprime, G, i, u)$
\State $\procAdjustEndpoint (\repbarprime, G, i+1, v)$
\State
\Statex \textbf{Validate:}
\If {$\eptg{\ppbar'}=G'$ and $\repbarprime$ satisfies $(P3)$}
\State \Return $\repbarprime$
\Else
\State \Return ``NO''
\EndIf
\State
\Function{\procAdjustEndpoint}{$\repbar, G, p, w$}
\Comment $w$ is the endpoint of $P_p$ to be adjusted
\State $e_w$ denotes the tail of $P_p$ incident to $w$
\State $\xx_w$ denotes $\set{P_x:  e_w \in P_x \textrm{~and~} \set{p,x} \notin E(G)}$
\State $\yy_w$ denotes $\set{P_y: P_p \cap P_y = \set{e_w} \textrm{~and~} \set{p,y} \in E(G)}$
\While {$\yy_w = \emptyset$ and $\xx_w \neq \emptyset$}
\State $tr(P_p,e_w)$
\EndWhile
\If {$\xx_w \neq \emptyset$}   \Comment Also $\yy_w \neq \emptyset$ as the while loop terminated
\State Subdivide $e_w$ into two edges $e_w, e_w'$ \Comment Revert the minifying operation
\For {$P_x \in \xx_w$}
\State $tr(P_w,e_{w'})$
\EndFor
\State $tr(P_p,e_w)$
\EndIf
\EndFunction
\end{algorithmic}
\end{algorithm}

\begin{figure}[htbp]
\centering
\includegraphics[width=\textwidth]{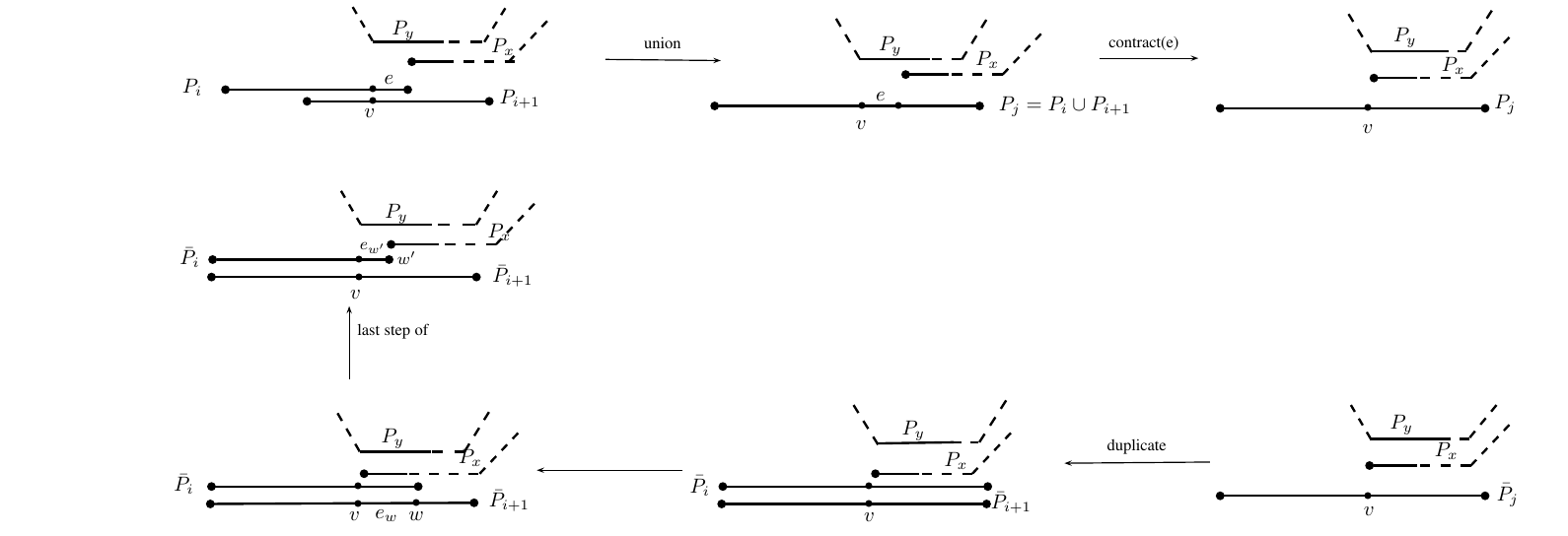}
\caption{The effect of union and minifying operations, and the reversal of this effect by Procedure $\procAdjustEndpoint$ (invoked with $p=i$).}\label{fig:AdjustEndpoint}
\end{figure}

\begin{theorem}\label{thm:FindMinimalRepresentation}
Instances of $\prbpthree$ satisfying property $(P2)$ can be solved in polynomial time. YES instances have a unique solution, and whenever $n \geq 6$ this solution is a broken planar tour.
\end{theorem}
\begin{proof}
If $n=5$ the result follows from Lemma \ref{lem:unique-c5}. If $(G',C')$ is a "NO" instance, then $\algPTwoPThree$ returns ``NO'' in the validation phase. Therefore, we assume that $n \geq 6$, and that $(G',C')$ is a "YES" instance, i.e. it has at least one representation satisfying $(P3)$. We will show that for any pair $(G',C')$ satisfying $(P2)$, and a minimal representation $\repprime$ of $(G',C')$ that satisfies $(P3)$, the representation $\repbarprime$ returned by $\algPTwoPThree$ is a broken planar tour and
\[
\repbarprime \cong \repprime \textrm{~and~} \repbarprime \lesssim \repprime.
\]

We will prove by induction on the number $k$ of contractable edges of $(G',C')$. If $k=0$ then $(G',C')$ is not contractable, therefore satisfies $(P1)$. In this case the algorithm invokes $\algbpt$ and the claim follows from Theorem \ref{thm:BuildPlanarTourCorrect}.

Otherwise $k>0$. We assume that the claim holds for $k-1$ and prove that it holds for $k$. As $(G',C')$ contains at least one contractable edge, one such edge $\set{i,i+1}$ is chosen arbitrarily by the algorithm and contracted. The resulting pair $(G,C)=\contract{(G',C')}{\set{i,i+1}}$ has the following properties:
\begin{itemize}
\item{Satisfies $(P2),(P3)$.} (By Lemma \ref{lem:ContractionPreservesP2P3})
\item{The number of contractable edges is $k-1$.}
\item{$\abs{V(G)} \geq 6$.} This is because $\abs{V(G)}=\abs{V(G')}-1$ and $\abs{V(G')}>6$. Indeed, if $\abs{V(G')}=6$, we have $k=0$ by Lemma \ref{lem:NoContractableC6}.
\end{itemize}

Therefore, $(G,C)$ satisfies the assumptions of the inductive hypothesis. Let $\repprime$ be a minimal representation of $(G',C')$ satisfying $(P3)$. Then $\contract{\repprime}{P_i,P_{i+1}}$ is a representation of $(G,C)=\contract{(G',C')}{\set{i,i+1}}$. By the inductive hypothesis, $\repbar$ is a broken planar tour that satisfies
\[
\repbar \cong \contract{\repprime}{P_i,P_{i+1}} \textrm{~and~} \repbar \lesssim \contract{\repprime}{P_i,P_{i+1}}.
\]
In other words $\repbar$ is obtained from $\repprime$ by replacing the two paths $P_i, P_{i+1}$ with the path $P_i \cup P_{i+1}$, then applying a (possibly empty) sequence of minifying operations. By Lemma \ref{lem:MinimizationOfBrokenPlanarTours}, these minifying operations are $contract(e)$ for a tail $e$ of $P_i \cap P_{i+1}$. In the Uncontract phase, $\algPTwoPThree$ performs a reversal of these transformations. See Figure \ref{fig:AdjustEndpoint} for the following discussion. One endpoint of each one of $P_i$ and $P_{i+1}$ is an endpoint of $P_i \cup P_{i+1}$. Therefore, one needs to determine only one endpoint of each one of $P_i$ and $P_{i+1}$. First $P_i \cap P_{i+1}$ is duplicated and the so obtained paths are called $P_i$, $P_{i+1}$.

For $p \in \set{i,i+1}$, let $w$ be the endpoint of $P_p$ to be adjusted. $e_w$ denotes the tail of $P_p$ incident to $w$. We denote by $\xx_w$ the set of paths containing $e_w$ such that vertices of $G'$ corresponding to these paths are not adjacent to $p$. We denote by $\yy_w$ the set of paths intersecting $P_p$ only on $e_w$ and whose corresponding vertices in $G'$ are adjacent to $p$. If $\yy_w$ is empty (that is, every path that intersects $P_p$ also intersects $P_p \setminus \set{e_w}$), $e_w$ can be safely removed from $P_p$ without losing intersections. If $\xx_w$ is non-empty this removal is a necessary operation. The algorithm performs these tail removals as long as they are necessary and safe. If at the end of this loop, $\xx_w$ is empty then we are done. Otherwise $\xx_w$ and $\yy_w$ are non-empty, then $e_w$ can not be safely removed from $P_p$. In this case $\procAdjustEndpoint$ subdivides $e_w$ (thus reversing the minifying operation $contract(e)$) and removes one tail from $P_p$ and one tail from every path $X \in \xx_w$, so that $P_p$ does not intersect $X$ but still intersects every path $Y \in \yy_w$.
\end{proof}

\section{Pairs $(G,C)$ Satisfying $(P3)$}\label{sec:K4P4}
In the previous section we relaxed assumption $(P1)$. In this section we relax assumption $(P2)$, i.e. we allow sub-pairs isomorphic to $\kp$. In Section \ref{subsec:k4p4} we investigate the basic properties of the representations of such sub-pairs, and characterize the representations of pairs $(G,C)$ with at most $6$ vertices. In Section \ref{subsec:K4P4IntersectionsAgressiveContraction} we show that in bigger cycles such pairs can intersect only in a particular way, and we define the aggressive contraction operation that transforms a pair $(G'',C'')$ with a $\kp$ to a pair $(G',C')$ with one less vertex and at least one $\kp$ less. Using these results, in Section \ref{subsec:K4P4Algorithm} we present an algorithm that finds the unique minimal representation of a given pair $(G,C)$ satisfying $(P3)$ with more than $6$ vertices.

We denote a set of $4$ vertices inducing a sub-pair isomorphic to $\kp$ as an ordered quadruple where the first vertex is one of the endpoints of the the induced $P_4$, the second vertex is its neighbor and so on. The quadruple $(p,q,r,s)$ is a $\kp$ of $(G,G')$ whenever $\set{p,q,r,s}$ induces a sub-pair $\kp$ of $(G,G')$ and $\set{p,q}, \set{q,r}, \set{r,s} $ are the edges of $G'$. Clearly, $(p,q,r,s)=(s,r,q,p)$.

\subsection{Representations of $\kp$ and Small Cycles}\label{subsec:k4p4}
In this section we investigate representations of induced $\kp$ pairs and characterize the unique minimal representations of $(G,C_n)$ pairs containing an induced $\kp$.
We start with Lemma \ref{lem:K4P4Representation} that characterizes representations of $\kp$ pairs. Using this lemma we prove Theorem \ref{thm:C5WithK4P4} that presents the unique minimal representation of $(G,C_5)$ pairs containing a $\kp$. Together with Lemma \ref{lem:unique-c5} this completes the characterization of all the $(G,C_5)$ pairs because a $(G,C_5)$ satisfies $(P3)$ vacuously. We continue by proving Lemma \ref{lem:K4P4RepresentationsInBigCycles} more properties of minimal representations of induced $\kp$ sub-pairs of pairs $(G,C)$ with at least $6$ vertices. Using these results we show that a $(G,C_6)$ satisfying $(P3)$ does not contain sub-pairs isomorphic to $\kp$.

\newcommand{\allintr}{\bigcap \pp_K}

\begin{lemma}\label{lem:K4P4Representation}
Let $K=(i,i+1,i+2,i+3)$ be a $\kp$, $\rep$ be a representation of $K$, and $\allintr \defined P_i \cap P_{i+1} \cap P_{i+2} \cap P_{i+3}$. There is a path $core(K)$ of $T$ with endpoints $u,v$ such that:\\
(i)~$\split(P_i,P_{i+2})=\set{u}$, $\split(P_{i+1},P_{i+3})=\set{v}$, $P_{i+1}$ (resp. $P_{i+2}$) does not cross $u$ (resp. $v$).\\
(ii)~$\emptyset \neq \allintr \subseteq \left( P_{i+1} \cap P_{i+2} \right) \subseteq core(K)$. In particular $u \neq v$.\\
(iii)~At least one of $P_i, P_{i+3}$ crosses both endpoints of $core(K)$ and $\emptyset \neq \split(P_i, P_{i+3}) \subseteq \set{u,v}$.\\
(iv)~$P_{i+1} \cup P_{i+2}$ crosses both endpoints of $core(K)$.\\
(v)~The removal of the edges of $P_{i+1} \cup P_{i+2}$ from $T$ disconnects $P_i$ from $P_{i+3}$.
\end{lemma}
\begin{proof}
(i) Assume, by way of contradiction, that $\abs{\split(P_i, P_{i+2})}=\set{w,w'}$  where $w$ and $w'$ are distinct vertices of $T$. As $P_{i+1} \sim P_i$ and $P_{i+1} \sim P_{i+2}$ we conclude that $P_{i+1} \subseteq p_T(w,w')$. Since $P_{i+3} \nsim P_{i+1}$, $P_{i+3}$ splits from $P_{i+1}$ in at least one vertex $w''$ that is an intermediate vertex of $p_T(w,w')$. Then $P_{i+3}$ splits from $P_{i+2}$ at $w''$ contradicting the fact that $\set{i+2,i+3}$ is an $\enpt$ edge. Therefore, $\abs{\split(P_i, P_{i+2})}=1$ and by symmetry, $\abs{\split(P_{i+1},P_{i+3})}=1$. Let $\split(P_i, P_{i+2}) = \set{u}$ and $\split(P_{i+1}, P_{i+3}) = \set{v}$. We define $core(K) = p_T(u,v)$. For the rest of the claim, assume by contradiction that $P_{i+1}$ crosses $u$. Then either $P_{i+1} \nsim P_i$ or $P_{i+1} \nsim P_{i+2}$, contradicting the the fact that $\set{i,i+1}$ and $\set{i+1,i+2}$ are $\enpt$ edges.

At this point we can uniquely define the following edges that will be used in the rest of the proof: $e_i$ (resp. $e_{i+2}$) is the edge of $P_i \setminus P_{i+2}$ (resp. $P_{i+2} \setminus P_i$) incident to $\split(P_i,P_{i+2})$ ($=\set{u}$). We define $e_{i+1}$ and $e_{i+3}$ similarly. Note that $e_i \neq e_{i+2}$ and $e_{i+1} \neq e_{i+3}$. 

(ii) A claw-clique of size $4$ contains exactly one $\enpt$ edge, however a path isomorphic to $P_4$ contains three edges. Therefore, the representation of $K_4$ is an edge-clique. Let $e$ be an edge defining this edge-clique, i.e. $e \in \allintr$. The removal of $e$ from $T$ disconnects it into two subtrees. In order to prove that $\allintr \subseteq core(K)$ it suffices to show that $u$ and $v$ are in different subtrees. Assume, by way of contradiction that $u,v$ are in the same subtree $T_r$ with root $r$ where $r$ is an endpoint of $e$. Let $r'$ be the least common ancestor of $u,v$ in $T_r$ (possibly $u=v$ in which case $r'=u=v$). All the $4$ paths contain $e$ and cross $r'$ (so that each one crosses at least one of $u,v$), i.e. they ``enter'' $r'$ from the same edge $e'$ (where possibly $r'=r$ and $e'=e$). If $r' \notin \set{u,v}$ then as $P_{i+1}$ crosses $v$ and $P_{i+2}$ crosses $u$, $r' \in \split(P_{i+1}, P_{i+2})$, contradicting $P_{i+1} \sim P_{i+2}$. Therefore, we can assume without loss of generality that $r'=u$. Then the edges $e_i$ and $e_{i+2}$ are incident to $r'$. Then $P_{i+1}$ (resp. $P_{i+3}$) contains $e_i$ (resp. $e_{i+2}$)  because $P_{i+1} \sim P_i$ (resp. $P_{i+3} \sim P_{i+2}$). Therefore, $r' \in \split(P_{i+1}, P_{i+2})$, contradicting $P_{i+1} \sim P_{i+2}$. Therefore, $u$ and $v$ are in different subtrees, i.e. $e \in p_T(u,v)=core(K)$. Since $e$ is chosen as an arbitrary edge defining the edge-clique this implies that $\allintr \subseteq core(K)$. It remains to prove that $P_{i+1} \cap P_{i+2} \subseteq core(K)$. For this purpose, it is sufficient to show that both of $P_{i+1}$ and $P_{i+2}$ have one endpoint in $core(K)$. Indeed, assume without loss of generality that $P_{i+1}$ does not have an endpoint in $core(K)$. Then $P_{i+1}$ crosses $u$, a contradiction.

Consult Figure \ref{fig:K4P4Representation} for the rest of the proof.

(iii) By the above discussion, $u$ (resp. $v$) is an intermediate vertex of $P_i$ and $P_{i+2}$ (resp. $P_{i+1}$ and $P_{i+3}$), and they all intersect in at least one edge $e \in core(K)$. In order to see the first part of the claim assume, by way of contradiction, that both of $P_i$ and $P_{i+3}$ have an endpoint in $core(K)$, in this case $\allintr$ is between these two endpoints. Therefore, $P_i \sim P_{i+3}$, a contradiction.

We now proceed to show the rest of the claim: Let $w \in \split(P_i,P_{i+3})$. $e_i \notin P_{i+3}$ because otherwise $P_{i+3} \nsim P_{i+2}$, and by symmetry $e_{i+3} \notin P_i$. Therefore, $w$ is on $core(K)$. On the other hand $w$ is not an intermediate vertex of $core(K)$. Indeed, consider the two sub-paths obtained by removing $e$ from $core(K)$. If $w$ is an intermediate vertex of $core(K)$, then at least one of $P_{i+3} \nsim P_{i+2}$, $P_i \nsim P_{i+1}$ holds, depending on the sub-path $w$ belongs. We conclude $w \in \set{u,v}$. Together with $P_i \nsim P_{i+3}$, this implies the claim. Note that $\split(P_i, P_{i+3})=\set{u,v}$ if and only if both of $P_i$ and $P_{i+3}$ cross both endpoints $u,v$ of $core(K)$.

(iv) As $\set{i+1,i+2}$ is an $\enpt$ edge, $Q \defined P_{i+1} \cup P_{i+2}$ is a path. Moreover, $e_{i+1} \in P_{i+1}$ and $e_{i+2} \in P_{i+2}$, therefore $\set{e_{i+1},e_{i+2}} \subseteq Q$, implying the claim.

(v) It suffices to show that $core(K)$ separates $P_i$ and $P_{i+3}$. Suppose that after the removal of $core(K)$ the two paths still intersect. This is possible only if $e_{i+3} \in P_i$ or $e_i \in P_{i+3}$. Assume without loss of generality that $e_{i+3} \in P_i$. Then $P_i \nsim P_{i+1}$, a contradiction.
\end{proof}

\begin{figure}[htbp]
\centering
\includegraphics[width=\textwidth]{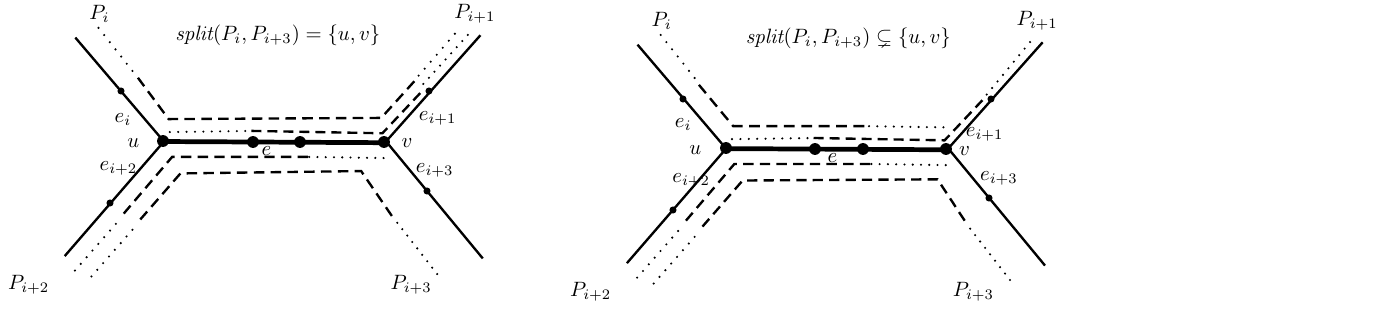}
\caption{Representations of $\kp$ pairs where $\split(P_i,P_i+3) = \set{u,v}$ and $\split(P_i,P_i+3) \subsetneq  \set{u,v}$, respectively.}\label{fig:K4P4Representation}
\end{figure}

Pairs $(G,C_5)$ with induced $\kp$ pairs turn out to be different than bigger cycles. One difference is that a $C_5$ contains a vertex incident to two distinct vertices of a $P_4$ whereas this clearly does not hold for a bigger cycle. Therefore, we analyse this case separately. We recall that a pair $(G,C_5)$ satisfies $(P3)$ vacuously, and that in Section \ref{subsec:c5c6} we found the unique minimal representation of a pair $(G,C_5)$ that satisfies $(P2)$. We now investigate the representation of a pair $(G,C_5)$ that does not satisfy $(P2)$.

\begin{theorem}\label{thm:C5WithK4P4}
If $(G,C_5)$ does not satisfy $(P2)$ then (i) $G$ is isomorphic to the graph depicted in Figure \ref{fig:C5WithK4P4}, and (ii) $(G,C_5)$ has a unique minimal representation also depicted in Figure \ref{fig:C5WithK4P4}.
\end{theorem}
\begin{proof}
Assume without loss of generality $K=(0,1,2,3)$ is a $\kp$ of $(G,C_5)$, and let $core(K) = p_T(u,v)$. If $\split(P_0,P_3)=\set{u,v}$ then $P_4 \subseteq core(K)$, implying that $P_4 \sim P_1$ or $P_4 \sim P_2$, i.e. at least one of $\set{1,4}$ or $\set{2,4}$ is an $\enpt$ edge, a contradiction.

Assume without loss of generality $\split(P_0,P_3)=\set{u}$, and that $P_3$ crosses both $u$ and $v$. Then $P_0$ has one endpoint $u'$ in $core(K)$, and $P_0 \cap P_3 = p_T(u,u')$.

As $P_4 \sim P_0$ and $P_4 \sim P_3$, we have $P_4 \sim p_T(u,u')$ and $P_4$ does cross $u$. Therefore, $P_4$ intersects $core(K)$. By Lemma \ref{lem:K4P4Representation} (iv) $core(K) \subseteq P_1 \cup P_2$. We conclude that $P_4 \cap (P_1 \cup P_2) \neq \emptyset$, i.e. $P_4 \cap P_1 \neq \emptyset$ or $P_4 \cap P_2 \neq \emptyset$. As $\set{4,1}$ and $\set{4,2}$ are not $\enpt$ edges, we have that $P_4 \nsim P_1$ or $P_4 \nsim P_2$. On the other hand $P_4$ does not cross $u$ and by Lemma \ref{lem:K4P4Representation} (i), $P_2$ does not cross $v$, thus $\split(P_2,P_4)=\emptyset$. Therefore, $P_4 \nsim P_1$ and $P_4 \parallel P_2$. Moreover, $\split(P_4,P_1)=\set{v}$, i.e. $P_4$ crosses $v$. Therefore, one endpoint $u''$ of $P_4$ is in $p_T(u,u')$, and must be between $u'$ and the endpoint of $P_1$ in $core(K)$.

It is easy to see that the path $\allintr$ can be contracted to one edge $e$ without affecting the relationships between the paths. Similarly, any edge between $u$ and $e$, and any edge between $e$ and $u''$ can be contracted. The path $p_T(u',u'')$ can be contracted to one edge, and the path $p_T(u',v)$ can be contracted to a single vertex $v$. This leads to the representation depicted in Figure \ref{fig:C5WithK4P4}.
\end{proof}

\begin{figure}[htbp]
\centering
\includegraphics{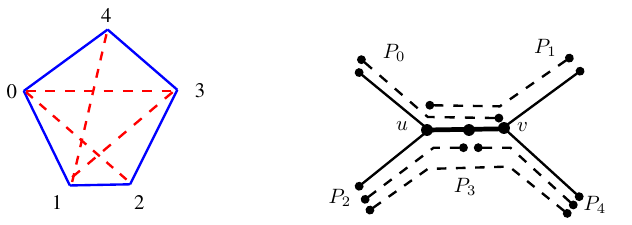}
\caption{The unique $(G,C_5)$ pair that does not satisfy $(P2)$ and its unique minimal representation.}\label{fig:C5WithK4P4}
\end{figure}

We now observe a property of the representations of $(G,C_5)$ in order to demonstrate the first family of non-$\enpt$ graphs.

\begin{lemma}
$G+C_5$ is not an $\enpt$ graph whenever $G$ is not a complete graph.
\end{lemma}
\begin{proof}
A pair $(G',C_5)$ satisfies $(P3)$ vacuously. If $(G',C_5)$ satisfies $(P2)$ then by Lemma \ref{lem:unique-c5}, its unique minimal representation is the one depicted in Figure \ref{fig:eptn-c5}. Otherwise, by Theorem \ref{thm:C5WithK4P4}, its unique minimal representation is the one depicted in Figure \ref{fig:C5WithK4P4}.
Let $i \in V(G)$. $i$ is adjacent to every vertex of $C_5$. We observe that in both cases above a) $P_i$ is a sub-path of $p_T(u,v)$, and b) there is a specific edge $e$ of $p_T(u,v)$ that is also in $P_i$. Therefore, for any two vertices $i,j \in V(G)$ $P_i$ and $P_j$ are intersecting sub-paths of $p_T(u,v)$, thus $P_i \sim P_j$. We conclude that $G$ is a complete graph.
\end{proof}

We now extend Lemma \ref{lem:K4P4Representation}. As opposed to Lemma \ref{lem:K4P4Representation} that investigates the structure of a $\kp$ regardless of any specific context, the following lemma provides us with more properties of minimal representations satisfying $(P3)$ of pairs $(G,C)$.
\begin{lemma}\label{lem:K4P4RepresentationsInBigCycles}
Let $K=(i,i+1,i+2,i+3)$ be a $\kp$ of a pair $(G,C)$ satisfying $(P3)$ on at least $6$ vertices. Let $\rep$ be a minimal representation of $(G,C)$ and let $\pp_K=\set{P_i: i\in K}$.\\
(i) $\allintr=\set{e}$ for some edge $e$ which is used exclusively by the paths of $\pp_K$, i.e. $e \in P_j \Rightarrow j \in K$.\\
(ii) $e$ divides $T$ into two subtrees $T_1, T_2$ such that $T_1$ is a cherry of $\left< T,\pp_K \right>$ with center $w_1$. We denote this subtree as $cherry(K)$.\\
(iii) $\split(P_i, P_{i+3})= \set{w_2} \subseteq V(T_2)$.\\
(iv) $N_G(j)=K$ if and only if $\split(P_j,P_i) \cup \split(P_j,P_{i+3})=\set{w_1}$. The unique vertex $j$ satisfying this condition is one of $i+1,i+2$.
\end{lemma}

\begin{proof}
Consult Figure \ref{fig:MinimalK4P4InACycle} for this proof.

(i) Let without loss of generality $i=0$. By Lemma \ref{lem:K4P4Representation}, $\allintr$ is not empty. By contradiction assume that a path $P_j \notin \pp_K$ intersects $\allintr$. Then $K \cup \set{j}$ is an edge-clique of $G$. We claim that this $K_5$ contains at least one red triangle, contradicting $(P3)$. Indeed, as $C$ has at least $6$ vertices, $j$ is adjacent in $C$ to at most one vertex $k \in \set{0,3}$. $K \setminus \set{k}$ contains one red edge. The endpoints of this edge together with $j$ constitute a red edge-clique. Therefore, no path of $\pp \setminus \pp_K$ intersects $\allintr$. Then no intermediate vertex of $\allintr$ is a split vertex. By the minimality of $\rep$, $\allintr$ consists of one edge, say $e$.

(ii) Let $T_1, T_2$ be the subtrees obtained by the removal of $e$ from $T$. As $V(G) \setminus K$ is a connected component of $G$, the union of the paths $\pp \setminus \pp_K$ is a subtree $T'$ of $T$. $T'$ is a subtree of $T_1$ or a subtree of $T_2$, because otherwise there is at least one path of $\pp \setminus \pp_K$ using $e$, contradicting (i). Without loss of generality let $T_2$ be the subtree containing $T'$, and $T_1$ be the subtree that intersects only paths of $\pp_K$. By Lemma \ref{lem:K4P4Representation} (ii), $T_1$ contains exactly one endpoint of $core(K)$. For $i \in \set{1,2}$, let $w_i$ be the endpoint of $core(K)$ that is in $T_i$. $w_1$ is the only split vertex in $T_1$ because it contains only paths of $\pp_K$. As the representation is minimal, there are no edges between $e$ and $w_1$, as otherwise they could be contracted. Any subtree of $T_1$ starting with an edge incident to $w_1$ can be contracted to one path because the subtree does not contain split vertices. Moreover, this path can be contracted to one edge, because all the paths entering the subtree intersect in its first edge. Since there are only two such subtrees, $T_1$ is isomorphic to a $P_3$ and $w_1$ is its center.

(iii) Assume that $\abs{\split(P_0, P_3)}=2$. Then by Lemma \ref{lem:K4P4Representation}, $\split(P_0, P_3)=\set{w_1,w_2}$, i.e. $w_1$ is an internal vertex of both $P_0$ and $P_3$. In this case, one can remove from $P_0$ its unique edge in $T_1$ without affecting the relationships between the paths. This contradicts the minimality of $\rep$. Indeed, a) any change in $T_1$ affects relationships between paths of $\pp_K$ only, b) $\allintr$ is not affected, therefore all the paths of $\pp_K$ still intersect, c) $\set{w_1} = \split(P_0, P_3) = \split(P_0,P_2)$ and $\set{w_2} = \split(P_1, P_3)$ hold after the tail removal.

Now assume that $\split(P_0, P_3)=\set{w_1}$. $P_0$ crosses $w_2$ because $\split(P_0,P_2)=\set{w_2}$. Then $P_3$ does not cross $w_2$. As $P_4 \sim P_3$, $P_4,P_2,P_0$ intersect in the last edge of $P_3$, and thus constitute a red edge-clique, contradicting $(P3)$. We conclude that $\split(P_0, P_3)=\set{w_2}$.

(iv) First assume $j \notin \set{i+1,i+2}$. Clearly, $N_G(j) \neq K$. Moreover, we have $\split(P_j,P_i) \cup \split(P_j,P_{i+3}) \neq \set{w_1}$. Indeed, if $j \notin K$ then $w_1$ is not a vertex of $P_j$ and if $j \in \set{i,i+3}$ the condition holds because (iii).

We now assume $j \in \set{i+1,i+2}$. By Lemma \ref{lem:K4P4Representation} (v), the removal of $P_1 \cup P_2$ disconnects $P_0$ from $P_3$. Then the tree $T'$ intersects $P_1 \cup P_2$. Therefore, at least one of $P_1, P_2$ intersects $T'$. By Lemma \ref{lem:K4P4Representation} i) one of $P_1, P_2$ does not cross $w_2$, i.e. does not intersect $T_2$ which in turn includes $T'$, a contradiction. We conclude that exactly one of $P_1, P_2$ intersects $T'$. In other words exactly one of $1,2$ is adjacent to $V(G) \setminus K$.  Assume $N_G(i+1)=K$. Then $P_{i+1} \cap T_1 \neq \emptyset$, therefore $\split(P_{i+1},P_i) = \emptyset$, i.e. $\split(P_{i+1},P_{i+3})=\set{w_1}$, concluding the claim. The case $N_G(i+2)=K$ is symmetric.
\end{proof}

\begin{figure}[htbp]
\centering
\includegraphics[width=\textwidth]{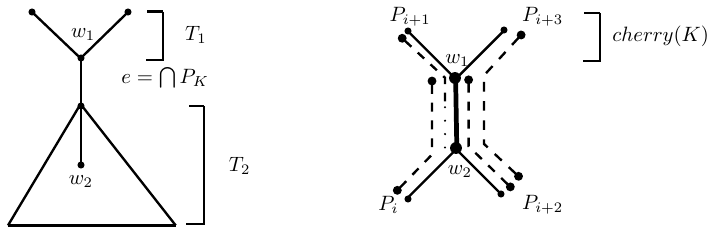}
\caption{A minimal representation of a pair $(G,C)$ with an induced $\kp$ with $N_G(i+1)=K$.}\label{fig:MinimalK4P4InACycle}
\end{figure}

We term, as \emph{isolated}, the vertex $j \in \set{i+1,i+2}$ of $K=(i,i+1,i+2,i+3)$ satisfying $N_G(j)=K$ whose existence and uniqueness are guaranteed by Lemma \ref{lem:K4P4RepresentationsInBigCycles} (iv). We recall that $(i,i+1,i+2,i+3)=(i+3,i+2,i+1,i)$, and in view of this result, we introduce an alternative notation: We denote $K$ as $[i,i+1,i+2,i+3]$ if $i+1$ is its isolated vertex, and as $[i+3,i+2,i+1,i]$ otherwise.

\begin{lemma}\label{lem:OtherPathsEnteringTheCore}
Let $K=[i,i+1,i+2,i+3]$ a $\kp$ of a pair $(G,C)$ with at least $6$ vertices, $\rep$ a minimal representation of $(G,C)$ satisfying $(P3)$. If there is a path $P_j \notin \pp_K$ intersecting $core(K)$, then $j=i-1$ and $\abs{core(K)}=2$, otherwise $\abs{core(K)}=1$.
\end{lemma}

\begin{proof}
Let $\allintr = \set{e}$, and assume that $j \notin K$ and $P_j \cap core(K) \neq \emptyset$. Recall that $e \notin P_j$. If $P_j$ splits from $core(K)$ then it splits from each one of $P_i, P_{i+2}, P_{i+3}$. In particular $\set{j,i,i+2}$ constitutes a red edge-clique, thus violating $(P3)$. If $P_j \subseteq core(K)$ then $P_j \sim P_{i+2}$ implying $j \in \set{i+1,i+3} \subset K$, contradicting our assumption. Therefore, $P_j$ crosses the endpoint $w_2$ of $core(K)$. Then $P_j$ intersects with each one of $P_i, P_{i+2}, P_{i+3}$ in the last edge of $core(K)$. Therefore, a) $P_j \nsim P_{i+2}$ because $j \notin \set{i+1,i+3}$, and b) $P_i \nsim P_{i+2}$. If $P_j \nsim P_i$ then $\set{j,i,i+2}$ constitutes a red edge-clique, violating $(P3)$. Therefore, $P_j \sim P_i$, implying $j=i-1$. Note that $P_{i+1} \cap P_{i-1} = \emptyset$ because $i+1$ is isolated. $P_{i-1} \cap core(K)$ consists of a single edge $e'(\neq e)$, because otherwise they can be contracted to a single edge without affecting the relationships between the paths $P_{i-1},P_i,P_{i+2},P_{i+3}$ that are the only paths that intersect the contracted edges. Then $core(K)$ consists of the two edges $e,e'$. If $P_{i-1}$ does not intersect $core(K)$ then $\pp_K$ are the only paths that intersect $core(K)$. Therefore, all the edges of $core(K)$ can be contracted to one edge.
\end{proof}

\begin{lemma}\label{lem:NoAggressiveContractableC6}
A pair $(G,C)$ with $6$ vertices satisfying $(P3)$ does not contain an induced $\kp$.
\end{lemma}
\begin{proof}
Assume without loss of generality that $[0,1,2,3]$ is a $\kp$ of $(G,C)$. Let $\rep$ be a representation of $(G,C)$ satisfying $(P3)$. For $i \in \set{0,3}$ let $T_i$ be the unique connected component of $T \setminus core(K)$ intersecting $P_i$. By Lemma \ref{lem:OtherPathsEnteringTheCore}, $P_4$ does not cross $w_2$. Therefore, $P_4$ is completely in $T_3$. As $P_4 \cap P_5 \neq \emptyset$, $P_5$ intersects $T_3$. If $P_5$ is completely in $T_3$ then $P_5 \parallel P_0$, otherwise $P_5 \nsim P_0$. Both cases contradict the fact that $\set{5,0}$ is an edge of $C$.
\end{proof}

\subsection{Intersection of $\kp$ pairs and Aggressive Contraction }\label{subsec:K4P4IntersectionsAgressiveContraction}
We now focus on pairs with at least $7$ vertices. We start by analyzing the intersection of their $\kp$ sub-pairs.

\begin{lemma}\label{lem:twok4p4s}
Let $(G,C)$ be a pair with at least $7$ vertices satisfying $(P3)$, and $K=[i,i+1,i+2,i+3]$ a $\kp$ of $(G,C)$. Then\\
(i) there is at most one $\kp$, $K' \neq K$ such that $E(C[K]) \cap E(C[K']) \neq \emptyset$ and if such a $\kp$ exists then $K' = [i+5,i+4,i+3,i+2]$ (and therefore $\set{i+2, i+4}$ is an edge of $G$),\\
(ii) if $\set{i+2 , i+4}$ is an edge of $G$ then $K'=[i+5,i+4,i+3,i+2]$ induces a $\kp$ of $(G,C)$.
\end{lemma}
\begin{proof}
Let without loss of generality $i=0$.

(i) Since $1$ is isolated, we have $1 \notin K'$. Therefore, if $E(C[K]) \cap E(C[K']) \neq \emptyset$ for some $\kp$ $K'$ then $E(C[K]) \cap E(C[K']) = \set{\set{2,i}}$, i.e. $K' = (2,3,4,5)$. Since $3$ is adjacent to $1$, $3$ is not isolated in $K'$. Therefore, $K'=[5,4,3,2]$.

(ii) Assume $\set{2,4}$ is an edge of $G$ and that, by way of contradiction, $K'=\set{2,3,4,5}$ is not a $\kp$. Consult Figure \ref{fig:K4P4PairCharacterization} for the following discussion. For $j \in \set{0,3}$ let $T_j$ be the connected component of $T \setminus core(K)$ intersecting $P_j$. As $P_4 \sim P_3$, Lemma \ref{lem:OtherPathsEnteringTheCore} implies that $P_4$ is completely in $T_3$. $P_4 \nsim P_2$, by our assumption. Let $w_3$ be the endpoint of $P_3$ in $T_3$ and $w_4$ be the split vertex of $P_2$ and $P_4$. Then $w_3 \in p_T(w_2,w_4)$ (possibly  $w_3=w_4$). $P_5$ does not intersect at least one of $P_2$ and $P_3$, because otherwise $K'$ is a $\kp$. Then it does not intersect $P_3$. The union of the paths $P_6, \ldots P_{n-1}$ constitutes a subtree $T'$ of $T$ that intersects both $P_0$ and $P_5$. Therefore, there is at least one path $P_j \in \set{P_6, \ldots P_{n-1}}$ crossing the last edge of $P_3$ (incident to $w_3$). Then $\set{2,4,j}$ is an edge-clique defined by this edge. Moreover, a) $P_2 \nsim P_4$, b) $P_j \nsim P_2$ because $j \notin \set{1,3}$, $P_j \nsim P_4$ because $j \notin \set{3,5}$. Therefore, $\set{2,4,j}$ is a red edge-clique, contradicting the assumption that $(P3)$ is satisfied.
\end{proof}

\begin{figure}[htbp]
\centering
\includegraphics{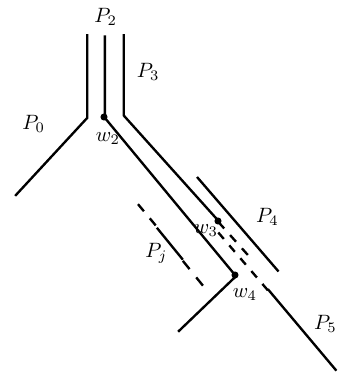}
\caption{Proof of Lemma \ref{lem:twok4p4s}.}\label{fig:K4P4PairCharacterization}
\end{figure}

By the above lemma $\kp$ sub-pairs may intersect only in pairs. We term two intersecting $\kp$ pairs as \emph{twins}, and a $\kp$ not intersecting with another as a \emph{single} $\kp$.

Given a $\kp$ $K=[i,i+1,i+2,i+3]$ of a pair $(G'',C'')$ satisfying $(P3)$, the \emph{aggressive contraction} operation is the replacement of the vertices $i+2,i+3$ by a single vertex $(i+2).(i+3)$. We denote by $(\contract{G''}{e},\contract{C''}{e})$ the resulting pair (where $e=\set{i+2,i+3}$) as $\aggressivegcdpk$. The following lemma characterizes the aggressive contraction operation in the representation domain.

\begin{lemma}\label{lem:aggressiveContraction}
Let $(G'',C'')$ be a pair with at least $7$ vertices, $\repdprime$ be a representation of it satisfying $(P3)$, and
$K=[i,i+1,i+2,i+3]$ be a $\kp$ of $(G'',C'')$. Then:\\
$\aggressivegcdpk$ is a pair satisfying $(P3)$ and a representation $\repprime$ of $(G',C') = \aggressivegcdpk$ satisfying $(P3)$ is obtained from $\repdprime$ by first removing $cherry(K)$ and also $cherry(K')$ if $K$ and $K'$ are twins, and then applying the union operation to $P_{i+2}$ and $P_{i+3}$.
\end{lemma}
\begin{proof}
Let without loss of generality $i=0$. Recall that by Lemma \ref{lem:twok4p4s}, $\set{2,4}$ is an edge of $G''$, if and only if $K$ is a twin. Figure \ref{fig:ACofOneK4P4} illustrates the following two steps in the case that $K$ is a single.

(Step $1$) We remove $cherry(K)$ (and also $cherry(K')$ when $K$ and $K'$ are twins) from $T''$. By Lemma \ref{lem:K4P4RepresentationsInBigCycles} we know that by removing cherries we don't lose any edge intersection, and we lose exactly one split vertex per cherry, namely the center of the cherry. This vertex (or vertices) is $\split(P_1, P_3)$ (and also $\split(P_2, P_4)$ when $K$ is a twin). Thus the edge $\set{1,3}$ (and also $\set{2,4}$ when $K$ is a twin) becomes blue. As no new red edges are introduced, the resulting representation does not contain red edge-cliques, i.e. satisfies $(P3)$.

(Step $2$) We contract the resulting graph on the edge $\set{2,3}$. We claim that this contraction is defined. Indeed assume by contradiction that $\set{2,3}$ participates in a $BBR$ triangle. This $BBR$ triangle is one of $\set{1,2,3}$ and $\set{2,3,4}$. Then one of $\set{1,3}$ and $\set{2,4}$ is a red edge, contradicting the fact that these edges (if exist) becomes blue after step $1$. This contraction corresponds to the union operation on the paths $P_2,P_3$, and by Lemma \ref{lem:ContractionPreservesP2P3} the resulting graph satisfies $(P3)$.
\end{proof}

\subsection{Algorithm}\label{subsec:K4P4Algorithm}
\newcommand{\algPThree}{\textsc{FindMinimalRepresentation-P3}}
\newcommand{\procSplitEar}{\textsc{MakeCherry}}

\begin{figure}[htbp]
\centering
\includegraphics{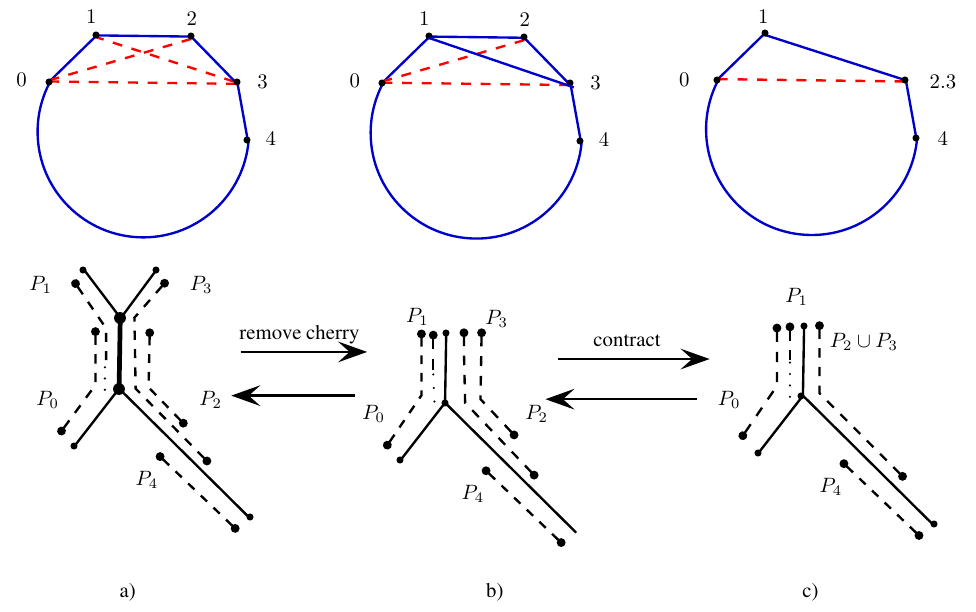}
\caption{Aggressive contraction of a single $\kp$.}\label{fig:ACofOneK4P4}
\end{figure}

\begin{figure}[htbp]
\centering
\includegraphics{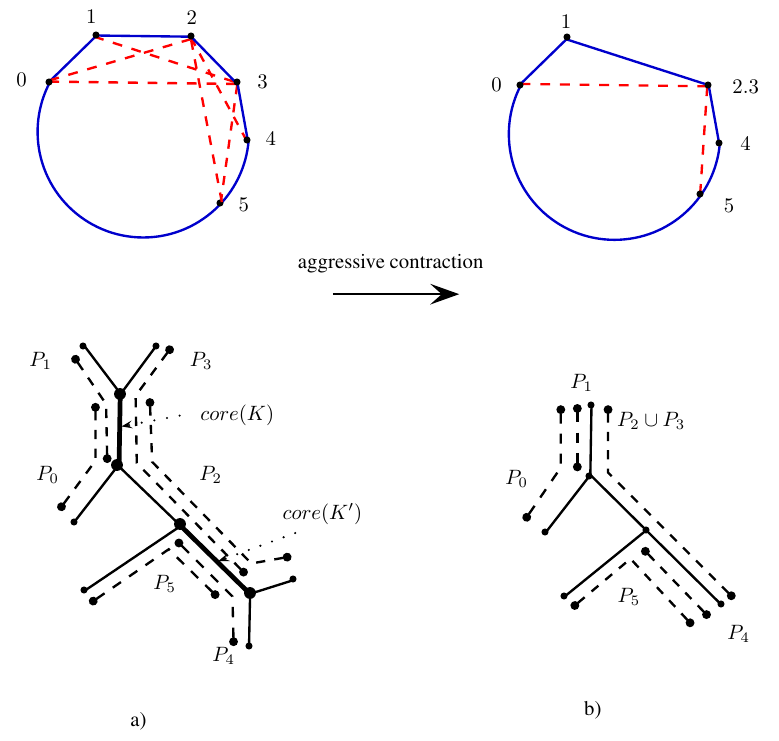}
\caption{Aggressive contraction of twins.}\label{fig:ACOfTwoK4P4s}
\end{figure}

Lemma \ref{lem:aggressiveContraction} implies an algorithm for finding the unique minimal representation of
pairs satisfying $(P3)$. Algorithm $\algPThree$ is a recursive algorithm that processes a single $\kp$ or a twin of $\kp$s at every invocation. The processing is done by applying aggressive contraction to convert the involved $\kp$(s) to $(K_3,P_3)$ (s),  solving the problem recursively, and finally transforming the representation of the $(K_3, P_3)$ to a representation of a $\kp$. In the Build Representation phase, Algorithm ~\algPThree~ performs the reversal of steps 1 and 2 described in Lemma \ref{lem:aggressiveContraction}, (see Figures \ref{fig:ACofOneK4P4}, \ref{fig:ACOfTwoK4P4s}).

\alglanguage{pseudocode}
\begin{algorithm}
\caption{$\algPThree(G'',C'')$}
\label{alg:algtwo}
\begin{algorithmic}[1]
\Require {$C''=\set{0,1,\ldots,\abs{V(G'')}-1}$ is an Hamiltonian cycle of $G''$ and $\abs{V(G'')} \geq 6$}
\Ensure {A minimal  representation $\repbar$ of $(G'',C'')$ satisfying $(P3)$ if any}
\If {$(G'',C'')$ is $\kp$-free}
\State \Return $\algPTwoPThree(G'',C'',\wdtgdprime)$
\EndIf
\State
\Statex \textbf{Aggressive Contraction:}
\State Pick a $\kp$, $K=[i,i+1,i+2,i+3]$ of $(G'',C'')$.
\Comment Renumber vertices if necessary.
\State $(G',C') \gets \aggressivegcdpk$.
\State
\Statex \textbf{Recurse:}
\State $\repbarprime \gets \algPThree(G',C')$.
\State
\Statex \textbf{Build Representation:}
\State $\repbar \gets \repbarprime$.
\State Replace $P_{(i+2).(i+3)}$ by two copies $P_{i+2}$ and $P_{i+3}$ of itself.
\If {$i+2$ is adjacent to $i+4$ in $G''$}
\State
\Comment $K'=[i+5,i+4,i+3,i+2]$ is the twin of $K$ in $(G'',C'')$
\State \procSplitEar($\repbar, i+4, i+2$).
\Else
\Comment $K$ is a single
\State $w \gets$ the endpoint of $P_{i+2}$ which is not in $core(K)$.
\State \procAdjustEndpoint $(\repbar, G'', P_{i+2}, w)$.
\EndIf
\State \procSplitEar($\repbar, i+1, i+3$).
\State
\Statex \textbf{Validate:}
\If {$\eptg{\ppbar}=G''$ and $\ppbar$ satisfies $(P3)$}
\State \Return $\repbar$
\Else
\State \Return ``NO''
\EndIf
\Function{\procSplitEar}{$\repbar,p,q$}
\State Let $v \in V(\bar{T})$ be the common endpoint of $P_p, P_q$.
\State Add two new vertices $v', v''$ and two edges $\set{v,v'}, \set{v,v''}$ to $\bar{T}$.
\State Extend $P_p$ so that the endpoint $v$ is moved to $v'$.
\State Extend $P_q$ so that the endpoint $v$ is moved to $v''$.
\EndFunction
\end{algorithmic}
\end{algorithm}

A \emph{broken tour with cherries} is a representation obtained by adding cherries to a broken tour. See Figure \ref{fig:eptn-P3-cycle-representation} for an example of a broken planar tour with cherries and the graph pair induced by it.

\begin{theorem}
$\prbpthree$ can be solved in polynomial time. YES instances have a unique solution, and whenever $n \geq 6$ this solution is a broken planar tour with cherries.
\end{theorem}

\begin{proof}
As the case $\abs{V(G'')} < 6$ is already solved, we will show that for any given pair $(G'',C'')$ with $\abs{V(G'')} \geq 6$, $\algPThree$ solves $\prbpthree$. If $(G'',C'')$ is a ''NO'' instance, then the instance has no representation satisfying $(P3)$. In this case then the algorithm returns ``NO'' at the validation phase. Therefore, we assume that $(G'',C'')$ is a ''YES'' instance, and prove the claim by induction on the number $k$ of induced $\kp$ pairs of $(G'',C'')$.

If $k=0$ then $(G'',C'')$ does not contain any $\kp$ pairs, therefore satisfies $(P2)$. In this case the algorithm invokes $\algPTwoPThree$ and the claim follows from Theorem \ref{thm:FindMinimalRepresentation}.

Otherwise $k>0$. We assume that the claim holds for any $k'<k$ and prove that it holds for $k$. In this case, as the pair contains at least one $\kp$, one such pair $K$ is chosen arbitrarily by the algorithm and aggressively contracted. The resulting pair $(G',C')=\aggressivegcdpk$ has the following properties:
\begin{itemize}
\item{Satisfies $(P3)$.} (By Lemma \ref{lem:aggressiveContraction})
\item{The number of $\kp$ pairs is strictly less than $k$.}
\item{$\abs{V(G')} \geq 6$.} This is because $\abs{V(G')}=\abs{V(G'')}-1$ and $\abs{V(G'')}>6$. Indeed, if $\abs{V(G'')}=6$, we have $k=0$ by Lemma \ref{lem:NoAggressiveContractableC6}.
\end{itemize}

Therefore, $(G',C')$ satisfies the assumptions of the inductive hypothesis. Then, $\repbarprime$ is the unique minimal representation of $\aggressivegcdpk$ satisfying $(P3)$. It remains to show that the representation $\repbarprime$ is obtained from the representation $\repbar$ returned by the algorithm, by applying the steps described in Lemma \ref{lem:aggressiveContraction}.

Let without loss of generality $K=[i,i+1,i+2,i+3]$. By Lemma \ref{lem:twok4p4s}, $K$ has a twin $K'=[i+5,i+4,i+3,i+3]$ if and only if $\set{i+2,i+4}$ is an edge of $G''$. The algorithm checks the existence of this edge and takes two different actions, accordingly.

If $K$ is not a twin then step 2, i.e. the union operation is reversed by breaking apart the path $P_{(i+2).(i+3)}$ into two paths $P_{i+2}$ and $P_{i+3}$. Then step 1 is reversed by invoking procedure $\procSplitEar$ (see Figure \ref{fig:ACofOneK4P4}).

If $K$ is a twin, then $cherry(K)$ and $cherry(K')$ are uniquely determined by Lemma \ref{lem:K4P4RepresentationsInBigCycles} (ii) and procedure $\procSplitEar$ acts accordingly. This determines all the endpoints of $P_i,P_{i+1},P_{i+2},P_{i+3},P_{i+4},P_{i+5}$ that are different from the representation $\repbarprime$ (see Figure \ref{fig:ACOfTwoK4P4s}).
\end{proof}

\begin{figure}[htbp]
\centering
\includegraphics[width=\textwidth]{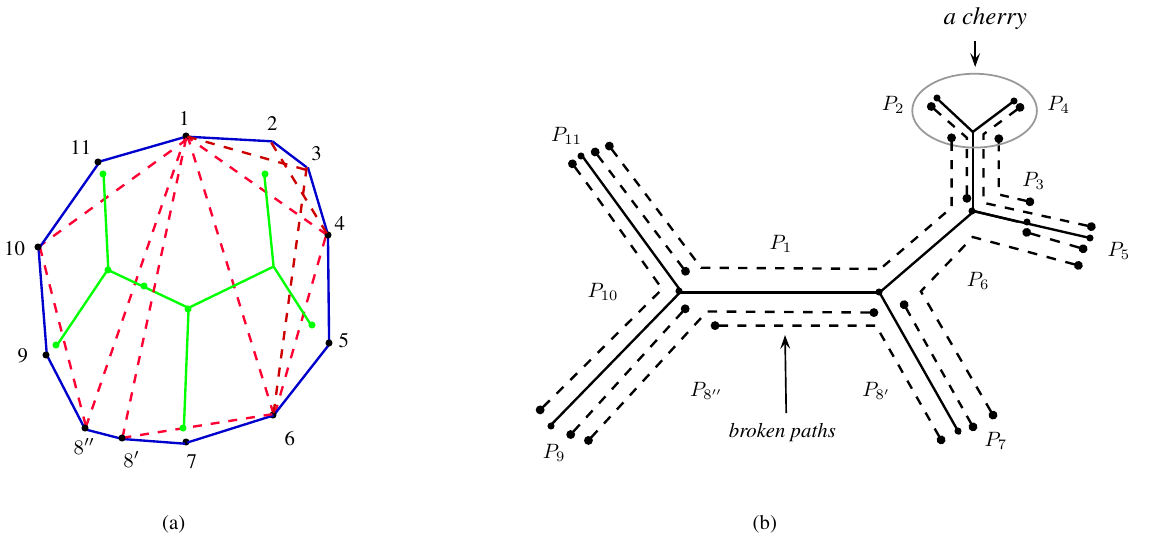}
\caption{(a) A pair ($G,C$) with a contractable edge $\set{8',8''}$ and a subgraph $\kp$ induced by $\set{1,2,3,4}$ (b) A representation of $(G,C)$: a broken planar tour with a cherry.}
\label{fig:eptn-P3-cycle-representation}
\end{figure}

\section{General Pairs $(G,C)$}\label{sec:General}
In this section we show that it is impossible to generalize the algorithms presented in the previous sections to the case where $(P3)$ does not hold, unless $\textsc{P}=\textsc{NP}$.

We start with a definition and a related lemma that are central to this section. Given a pair $(G,G')$ and a subset $S$ of $V(G)$, the \emph{component graph} $\comp(G,G',S)$ is a graph whose vertices correspond to the connected components $G_1, G_2, \ldots$ of $G \setminus S$ and two vertices corresponding to components $G_i, G_j$ are connected by an edge if and only if there is a vertex $v \in S$ adjacent to both of $G_i$ and $G_j$ in $G'$ (see Figure \ref{fig:ReductionPair} for an example). Whenever $G'$ is a cycle, we term a connected component of $G' \setminus S$ an \emph{arc} of $G'$ separated by $S$. Clearly, whenever $\abs{S} \geq 2$ every arc is adjacent to exactly $2$ vertices of $S$.

\begin{lemma}\label{lem:ComponentGraph3Colorable}
Let $(G,C)$ be a pair where $C$ is a Hamiltonian cycle of $G$, and $K$ be a maximal clique of $G \setminus C$. If there is a representation $\rep$ of $G$ where $\Delta(T) \leq 3$, then $\comp(G,C,K)$ is $3$-colorable.
\end{lemma}
\begin{proof}
If $\abs{K} \leq 3$, $G \setminus K$ has at most $3$ connected components, thus $\comp(G,C,K)$ is $3$-colorable. Therefore, we assume $\abs{K} > 3$. If $K$ is an edge-clique defined by an edge $e$ then the paths $\pp_K=\set{P_v: v \in K}$ are exactly the paths in $\pp$ that contain $e$. The edge $e$ divides $T$ into two subtrees $T_1,T_2$ rooted at the endpoints $r_1,r_2$ of $e$. Similarly, if $K$ is a claw-clique defined by a claw $\set{e_1,e_2,e_3}$, as $T$ has maximum degree $3$, the claw divides the tree into three subtrees $T_1, T_2, T_3$, rooted at the center $r_1=r_2=r_3=r$ of the claw. In both cases the following two statements hold: a) every path of $\pp \setminus \pp_K$ is contained in one of these subtrees, b) every path of $\pp_K$ that intersects a subtree $T_i$ crosses its root $r_i$.

All the vertices of a connected component $G_i$ are represented by paths that are in the same subtree $T_j$ ($j \in \set{1,2,3}$). This is because otherwise there are at least two adjacent vertices in $G_i$ that are in two different subtrees, a contradiction. We color every vertex $G_i$ of $\comp(G,C,K)$ with color $j \in \set{1,2,3}$ depending on the subtree on which the paths representing its vertices reside. It remains to show that if two connected components are adjacent in $\comp(G,C,K)$ they are colored with different colors.

Assume by contradiction that two components $G_1, G_2$ of $G \setminus K$ which are adjacent in $\comp(G,C,K)$ are colored with the same color $i$. Then, there is a vertex $v \in K$ and two vertices $v_1 \in G_1, v_2 \in G_2$ adjacent to $v$ in $C$. Moreover, $v_1$ and $v_2$ are not adjacent in $G$, because they are in different connected components. Therefore, (i) $P_v \sim P_{v_1}, P_v \sim P_{v_2}$, (ii) $P_{v_1} \parallel P_{v_2}$, (iii) $P_{v_1}$ and $P_{v_2}$ are in $T_i$, (iv) $P_v$ intersects $T_i$ and crosses its root $r_i$. Furthermore, we assume without loss of generality that $P_{v_1}$ is closer to $r_i$ than $P_{v_2}$ (see Figure \ref{fig:AdjacentSegmentsOfTheSameTree}). Consider the subtree $T'= \cup_{u \in G_2} P_u$ of $T_i$. $P_{v_1} \cap T' = \emptyset$, because otherwise there is a path $P_u$ representing a vertex $u \in G_2$ that intersects $P_{v_1}$, in other words $u \in G_2$ is adjacent to $v_1 \in G_1$, a contradiction.
Let $\set{v,v'}$ be the vertices of $K$ adjacent to the arc $v_2$ belongs to. $P_{v'}$ intersects $T_i$ and crosses its root $r_i$. Moreover, $P_{v'}$ intersects $T'$, as it is adjacent to at least one vertex of $G_2$. We conclude that $P_{v'}$ contains $P_{v_1}$. Then $v_1 \sim v'$, i.e. $v'$ and $v_1$ are adjacent in $C$. Therefore, $K=\set{v,v'}$, contradicting $\abs{K} > 3$.
\end{proof}

\begin{figure}[htbp]
\centering
\includegraphics{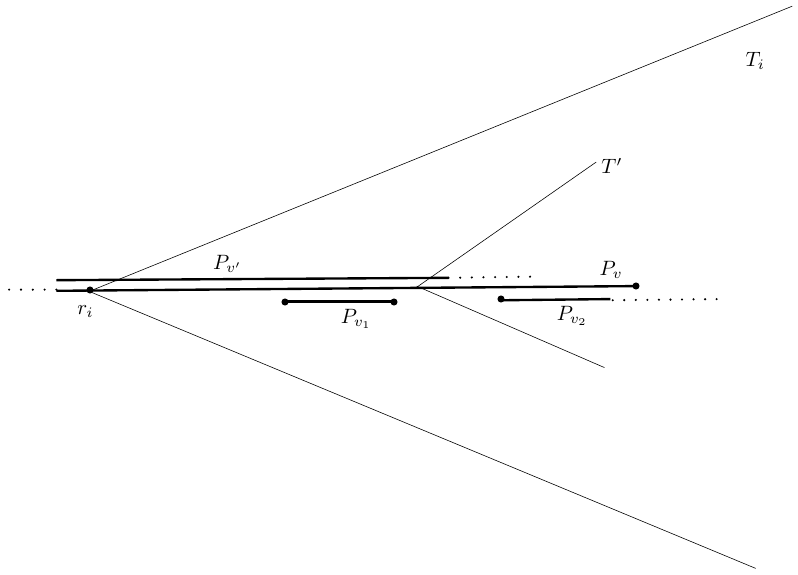}
\caption{Proof of Lemma \ref{lem:ComponentGraph3Colorable}.}\label{fig:AdjacentSegmentsOfTheSameTree}
\end{figure}

\begin{lemma}\label{lem:Degree3NPH}
It is $\nph$ to determine whether a given pair $(G,C)$ where $C$ is a Hamiltonian cycle of $G$ has representation $\rep$ with $\Delta(T) \leq 3$.
\end{lemma}
\begin{proof}
The proof is by reduction from the $3$-colorability problem. Given a graph $H$, we transform it to a pair $(G,C)$ such that $(G,C)$ has a representation on a tree with maximum degree $3$ if and only if $H$ is $3$-colorable.

Consult Figure \ref{fig:ReductionPair} for the following construction. Let $V(H)=\set{v_0, \ldots, v_{n-1}}$, $E(H)=\set{e_0, \ldots, e_{m-1}}$, and let $d_i=d_H(v_i)$.
The pair $(G,C)$ consists of $6m$ vertices. For every edge $e_k=\set{v_i,v_j}$ we construct a path $S_k=(u_{i,k}-u'_{i,k}-u_{j,k}-u'_{j,k}-u_k-u'_k)$ with $6$ vertices. The graph $C$ is a cycle obtained by concatenating these $m$ paths, in the order $S_0,S_1,\ldots,S_{m-1},S_0$, i.e. $u'_k$ is connected to $u_{i',k+1}$ where $e_{k+1}=\set{v_{i'},v_{j'}}$. $K$ is a clique of all the vertices in the even positions of the paths, i.e. $K=\set{u'_{i,k},u'_k : 0 < k < m, i \in e_k}$ (most of the edges induced by $K$ are not shown in the figure). For every $i < n$, $Q_i$ is a path $(u_{i,k_1}-\cdots-u_{i,k_{d_i}})$ where $e_{k_1},\ldots,e_{k_{d_i}}$ are the edges incident to $v_i$ in $H$. The set $E^{KQ}_i$ of edges connects vertices of $Q_i$ with vertices of $K$. Specifically, $E^{KQ}_i=\set{\set{u_{i,k_j},u'_{i,k_{j'}}}~|~1 \leq j' < j \leq d_i}$. Finally, $G=C \cup K \cup \left(\cup_{i<n} Q_i \right)  \cup \left( \cup_{i<n} E^{KQ}_i \right)$.

We claim that the vertices of the graph $H'=\comp(G,C,K)$ can be partitioned into two sets $A,B$ such that a) $H'[A]$ is isomorphic to $H$, b) $H'[B]$ is an independent set, c) $d_{H'}(v) \leq 2$ for every vertex $v \in B$. Indeed, $G \setminus K$ contains the vertices $\set{u_{i,k},u_k : k < m, i \in e_k}$ where each $u_k$ is an isolated vertex and the rest is the disjoint union of the paths $Q_i$. Therefore, the component graph $H'$ consists of the vertices $A=\set{Q_i : i<n}$ and $B=\set{u_k: k < m}$. For two vertices $v_i,v_j$ of $H$, $Q_i$ and $Q_j$ are connected by the vertex $u'_{i,k} \in V(K)$ if and only if $e_k=\set{v_i,v_j}$ is an edge of $H$. Therefore, $H'[A]$ is isomorphic to $H$. Moreover, a vertex $u_k$ of $G$ is connected to at most two paths $Q_i$ via its two neighbors in $C$. Therefore, $H'[B]$ is an independent set and every vertex of $B$ has degree at most $2$ in $H'$. We conclude that $H$ is $3$-colorable if and only if $H'$ is $3$-colorable. If $(G,C)$ has a representation $\rep$ with $\Delta(T) \leq 3$ then, by Lemma \ref{lem:ComponentGraph3Colorable}, $H'$ is $3$-colorable. It remains to show that if $H'$ is $3$-colorable then $(G,C)$ has such a representation. Given a $3$-coloring of $H'$, in the sequel we present such a representation $\rep$ (see Figure \ref{fig:ReductionRepresentation}).

We start with the construction of the tree $T$. $T$ has a vertex $r$ of degree at most $3$ that divides it into at most $3$ subtrees $T_1, T_2, T_3$, each of which with maximum degree $3$. Each $T_i$ corresponds to one color of the given $3$-coloring of $H'$. We describe in detail the subtree $T_1$, assuming without loss of generality that the vertices of $H'$ colored with color $1$ are $Q_1,Q_2,\ldots,Q_{n'}$ and $u_1,u_2,\ldots,u_{m'}$. $T_1$ contains a path $(r-e_1-\cdots-e_{m'}-v_1-\cdots-v_{n'})$. Each vertex $e_k$ starts a path $(e_k-\ell_k)$ of length $1$. Each vertex $v_i$ starts a path $(v_i-w_i-w_{i,k_1}-\cdots-w_{i,k_{d_i}})$ where $e_{k_1},\ldots,e_{k_{d_i}}$ are the edges incident to $v_i$ in $G$. Each vertex $w_{i,k}$ starts a path $(w_{i,k}-\ell_{i,k})$ of length $1$.

We proceed with the construction of the paths $\pp$. Every vertex $u_k$ of $G$ is represented by a path $P_k$ of length $1$ starting at vertex $\ell_k$. Each vertex $u_{i,k}$ of $G$ is represented by a path $P_{i,k}$ of length $3$ starting at $\ell_{i,k}$ and towards $r$. It remains to describe the representation of the vertices of $K$. Every vertex $u'$ of $K$ is adjacent to two vertices of $V(C) \setminus K$ in $C$. We represent $u'$ by a path between two leaves of $T$ (not all of them shown in the figure). These leaves are exactly the leaves that constitute endpoints of the paths corresponding to the two neighbors of $u'$. Specifically:
\begin{itemize}
\item{}
A vertex $u'_{i,k}$ of $S_k$ that is between two vertices $u_{i,k}$ and $u_{j,k}$ of $S_k$ is represented by a path $P'_{i,k}$ between the two leaves $\ell_{i,k}$ and $\ell_{j,k}$.
\item{}
A vertex $u'_{j,k}$ of $S_k$ that is between two vertices $u_{j,k}$ and $u_k$ of $S_k$ is represented by a path $P'_{j,k}$ between the two leaves $\ell_{j,k}$ and $\ell_k$.
\item{}
A vertex $u'_k$ of $S_k$ that is between two vertices $u_k$ of $S_k$ and $u_{i,k+1}$ of $S_{k+1}$ is represented by a path $P'_k$ between the two leaves $\ell_k$ and $\ell_{i,k+1}$.
\end{itemize}

The vertices $u_{i,k}$ and $u_{j,k}$ are in the connected components $Q_i$ and $Q_j$ respectively, which in turn are adjacent in $H'$ (by the existence of $u'_{i,k} \in K$ between them). They are therefore assigned different colors, i.e. the leaves $\ell_{i,k}$ and $\ell_{j,k}$ are in different subtrees of $T$. Therefore, $P'_{i,k}$ crosses $r$. It can be verified that this holds for the other two cases too. We conclude that the vertices of $K$ are represented by paths that cross $r$. If $H'$ is $2$-colorable then they constitute an edge-clique, otherwise they constitute a claw-clique. We leave to the reader to verify that $\rep$ is a representation of $(G,C)$.
\end{proof}
\begin{figure}[htbp]
\centering
\includegraphics[width=\textwidth]{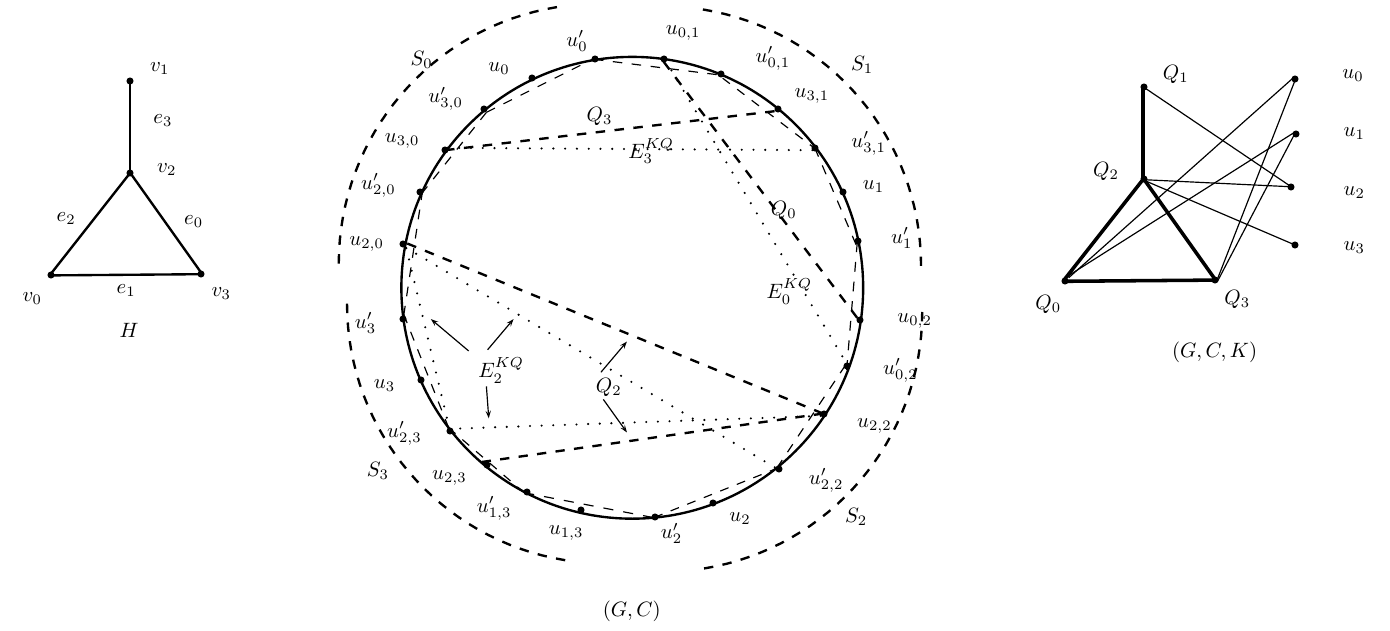}
\caption{A graph $H$, the corresponding pair $(G,C)$ and the component graph $\comp(G,C,K)$ where $K=\set{u'_{i,k},u'_k : 0 < k < m, i \in e_k}$.}
\label{fig:ReductionPair}
\end{figure}
\begin{figure}[htbp]
\centering
\includegraphics[width=\textwidth]{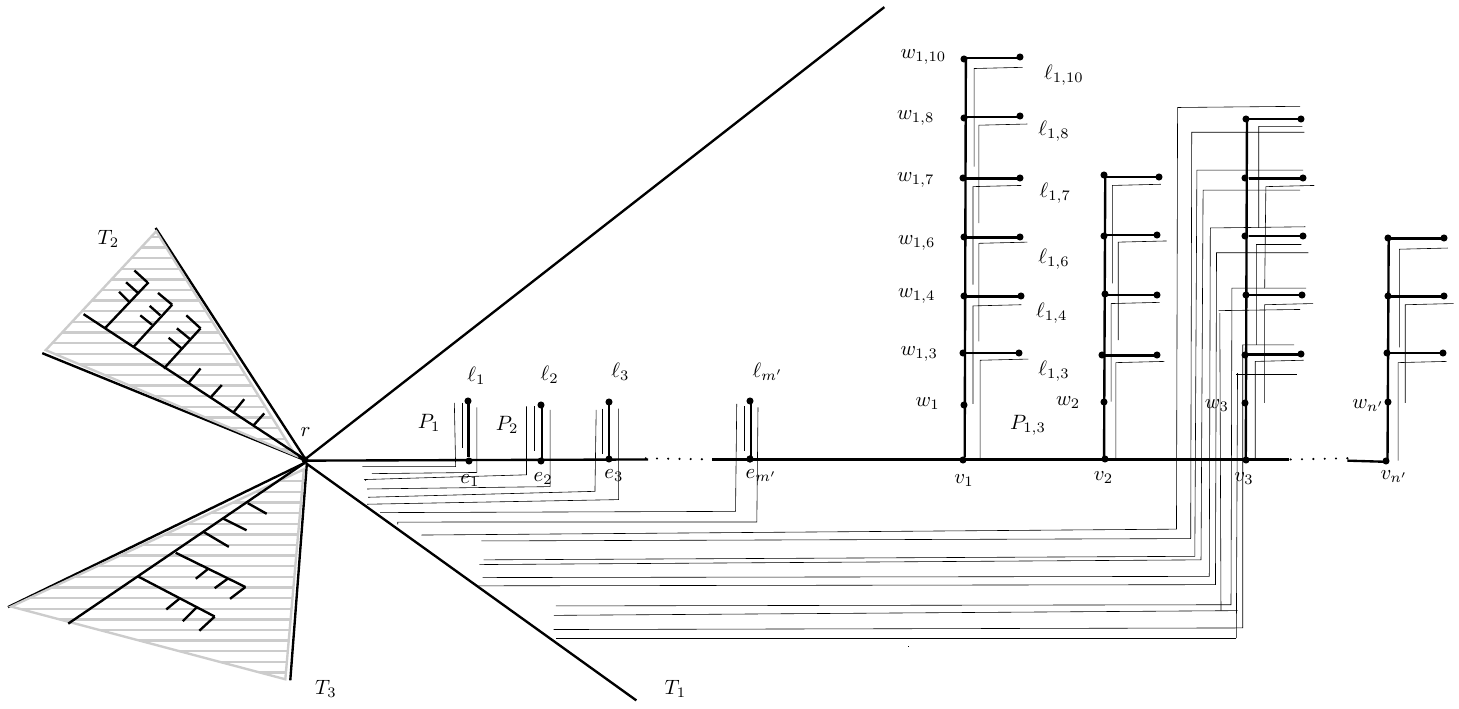}
\caption{A representation $\rep$ of a pair $(G,C)$ corresponding to some $3$-colorable graph $H$.}
\label{fig:ReductionRepresentation}
\end{figure}

\begin{theorem}
$\prb$ is $\nph$.
\end{theorem}
\begin{proof}
We claim that the decision version of the problem is $\nph$ even when $G$ is restricted to the family of $\vpt$ graphs. If the instance is a ``YES" instance, then $G$ is both a $\vpt$ and an $\ept$ graph. In this case, by Theorem 2 of \cite{Golumbic1985151}, $(G,C)$ has a representation on a tree with maximum degree $3$. If the instance is a ``NO" instance then, clearly, $(G,C)$ does not have a representation on a tree with maximum degree $3$. By Lemma \ref{lem:Degree3NPH} it is $\nph$ to decide whether $(G,C)$ has a representation on a tree with maximum degree $3$.
\end{proof}

\section{Conclusions and Future Work}\label{sec:conclusion}
In this study, we considered the characterization of minimal representations of $\enpt$ cycles. We described an algorithm finding the unique minimal representation of a pair of $\enpt$ and $\ept$ graphs that satisfy assumption $(P3)$, i.e. every red clique is represented by a claw-clique. Through this algorithm we characterized the representations of $\enpt$ cycles as broken planar tours with cherries. We have shown that there is no efficient algorithm to achieve this goal in general (i.e. without this assumption) unless $\textsc{P}=\textsc{NP}$.

Note that if we allow red edge-cliques, the representation is not necessarily a planar tour. The first such representation is a non-planar tour whose $\enpt$ graph is a cycle. Another example is depicted in Figure \ref{fig:eptn-monster-cycle}. This representation is not a tour since the set of its ``long'' paths does not define a cyclic permutation of the leaves of the tree.

Another direction of research would be to investigate the relation of $\enpt$ graphs with other graph classes. It is easy to see that $\enpt \setminus \ept \neq \emptyset$; for example consider the wheel $W_{5,1}=C_5+K_1$: it is not $\ept$ graph but is an $\enpt$ graph.
\cite{Golumbic1985151} characterized the graphs in $\vpt \cap \ept$. Another interesting research topic could be the characterization of the graphs in $\ept \cap \enpt$.

Last but not least, restriction to $\ept$ graphs of decision/optimization problems known to be $\nph$ in general graphs, such as minimum vertex coloring, maximum stable set, and hardness of recognition of $\enpt$ graphs seem to be major problems to investigate on these graphs.

\begin{figure}[htbp]
\centering
\includegraphics{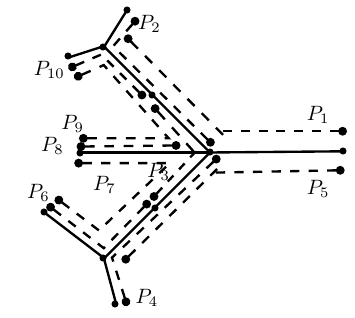}
\caption{A representation of $C_{10}$ which is not a tour.}
\label{fig:eptn-monster-cycle}
\end{figure}

\small

\bibliographystyle{abbrvnat}
\bibliography{GraphTheory,Optical,Mordo}
\label{sec:biblio}

\normalsize

\end{document}